\documentclass[letterpaper, 10pt]{amsart}
\usepackage{amsthm,amsmath,amssymb,mathrsfs,setspace,graphicx,mathabx,tcolorbox,color,tikz,enumitem,hyperref,cancel,graphics,bm,fullpage,esvect,multicol,hyperref}
\usepackage{subcaption} 
\usepackage{forest}
\usepackage{blkarray}
\usepackage[font=footnotesize,labelfont=bf]{caption}
\usepackage[utf8]{inputenc}
\usepackage{pst-node}
\usepackage{tikz-cd}
\usepackage[normalem]{ulem}
\usepackage{etoolbox}
\AtBeginEnvironment{align}{\setcounter{equation}{0}}

\newcommand{\excise}[1]{}

\definecolor{light}{gray}{.75}
\definecolor{med}{gray}{.5}
\definecolor{dark}{gray}{.25}

\newtheorem{theorem}{Theorem}
\newtheorem*{theorem*}{Theorem}
\numberwithin{theorem}{section}
\newtheorem{proposition}[theorem]{Proposition}

\newtheorem{corollary}[theorem]{Corollary}
\newtheorem{lemma}[theorem]{Lemma}
\newtheorem{conjecture}[theorem]{Conjecture}
\newtheorem{question}[theorem]{Question}
\theoremstyle{definition}
\newtheorem{definition}[theorem]{Definition}
\newtheorem{remark}[theorem]{Remark}
\newtheorem{example}[theorem]{Example}

\newcommand{\Cc}{{\mathbb C}}
\newcommand{\Nn}{{\mathbb N}}
\newcommand{\Qq}{{\mathbb Q}}
\newcommand{\Rr}{{\mathbb R}}
\newcommand{\Zz}{{\mathbb Z}}
\newcommand{\Pp}{{\mathbb P}}

\newcommand{\Kk}{\mathbb K}

\renewcommand{\AA}{\mathcal{A}}
\newcommand{\FF}{\mathcal{F}}

\newcommand{\II}{\mathcal{I}}

\newcommand{\cs}{\mathcal{S}}
\newcommand{\cn}{\mathcal{N}}
\newcommand{\ct}{\mathcal{T}}
\newcommand{\cg}{\mathcal{G}}

\newcommand{\bfa}{{\bf a}}
\newcommand{\bfg}{{\bf g}}
\newcommand{\bfh}{{\bf h}}
\newcommand{\bfc}{{\bf c}}

\newcommand{\isom}{\cong}				
\newcommand{\sm}{\setminus}				

\newcommand{\zmodtwo}{\Zz/2\Zz}

\DeclareMathOperator{\interior}{int}
\DeclareMathOperator{\MIN}{MIN}

\title{Invariants for level-1 phylogenetic networks under the Cavendar-Farris-Neyman Model}
\author{Joseph Cummings, Benjamin Hollering, Christopher Manon}
\date{March 2020}

\begin{document}

\begin{abstract}
Phylogenetic networks can model more complicated evolutionary phenomena that trees fail to capture such as horizontal gene transfer and hybridization. The same Markov models that are used to model evolution on trees can also be extended to networks and similar questions, such as the identifiability of the network parameter or the invariants of the model, can be asked. In this paper we focus on finding the invariants of the Cavendar-Farris-Neyman (CFN) model on level-1 phylogenetic networks. We do this by reducing the problem to finding invariants of \emph{sunlet networks}, which are level-1 networks consisting of a single cycle with leaves at each vertex. We then determine all quadratic invariants in the sunlet network ideal which we conjecture generate the full ideal. 
\end{abstract}

\maketitle

\section{Introduction}
The field of phylogenetics aims to determine the evolutionary relationships between species which are often represented with trees. There are some evolutionary phenomena that trees are unable to capture though. Non-treelike evolutionary processes include horizontal gene transfer where genetic material is passed laterally within a generation or hybridization \cite{M97, S94}. \emph{Phylogenetic networks} have emerged as a tool to model events in the evolutionary history of organisms that tree models are unable to represent. This has spurred an effort to study networks and develop methods to reconstruct them from data. Many results have already been obtained on the combinatorial properties of networks and many current methods for constructing networks are combinatorial in nature \cite{HS10, Ste16}. Other methods that have been used to infer trees have also been extended to networks such as maximum parsimony \cite{JNST06-Parsimony}, maximum likelihood \cite{JNST06-Likelihood}, and neighbor joining \cite{BM04}.

Recently, there has been work on the algebraic structure of network models motivated by the advances that algebraic methods achieved for tree models which include many identifiability results \cite{AR09, APRS10, BBND+19, LS15-ID, RS12} and descriptions of the \emph{phylogenetic invariants} of many tree-based models \cite{AR08, CS21, DK09, LS15-Strand, SS05}. Algebraic methods have also led to competitive methods for reconstructing trees such as those described in \cite{CK14, Eri05, FC15} which all utilize invariants. Gross and Long began the study of the algebraic and geometric structure of network models in \cite{GL18} and obtained some identifiability results for a certain class of Jukes-Cantor (JC) network models. Further identifiability results have since been obtained for networks using algebraic and combinatorial methods. These include level-1 networks under the coalescent model \cite{Ban19}, \emph{large-cycle} networks under the Kimura 2-Parameter (K2P) and Kimura 3-Parameter (K3P) models \cite{HS21}, and  level-1 networks under the JC, K2P, and K3P  models \cite{GI20}. There have also been some results obtained on the invariants of network models such as those in \cite{CF20}.

In this paper we focus on finding the invariants of the Cavendar-Farris-Neyman (CFN) model on level-1 phylogenetic networks. The discrete Fourier transform, which is used to simplify the parameterization of group-based models, such as the CFN model, can also be applied to network models as well \cite{GL18}. After applying this transform, CFN tree models become \emph{toric varieties} but the same is not true for CFN network models which makes analyzing their algebraic structure more difficult. As observed in \cite{GL18}, the toric fiber product of \cite{Sul07} can still be applied to group-based network models. Our approach leverages this toric fiber product structure to reduce the problem to that of finding the invariants for \emph{sunlet networks} which consist of only a single cycle. While sunlet network varieties are still not toric, they do have a lower-dimensional torus action on them meaning they are \emph{T-varieties} \cite{HHW19}. We use this torus action to break up the ideal of invariants of a $n$-leaf sunlet network into homogeneous graded pieces we call \emph{gloves}. As a result, we arrive at the following theorem.

\begin{theorem*}
A quadratic \(f\) is an invariant of the \(n\)-sunlet network if and only if it is an invariant for both of the underlying trees obtained by deleting a reticulation edge.
\end{theorem*}

We then explicitly produce all quadratic generators of the sunlet network ideal that lie in a given graded piece which gives a complete set of quadratic generators of the sunlet network ideal under the CFN model. We conjecture that the sunlet network ideal is generated by quadratics which would imply our set of quadratic generators actually generate the entire ideal. 

We have also studied the 4- and 5-leaf sunlet networks in more detail. We have shown through explicit computation that their corresponding varieties are normal and Gorenstein. This means that any level-1 network that can be built by gluing together 4- and 5-leaf sunlets along trees is normal and Cohen-Macaulay since these properties are preserved by the toric fiber product. Level-1 networks built from gluing 4- and 5-sunlets along leaves that are not adjacent to the reticulation vertex of the respective networks are also Gorenstein for the same reason but this may not hold if networks are glued together along leaves adjacent to the reticulation vertex. Lastly, we compute the multigraded Hilbert function of the 4-leaf sunlet network. All of these computational results along with an implementation of our algorithm to find quadratic generators and computational evidence for our conjectures can be found at:



\begin{center}
    \texttt{https://github.com/bkholler/CFN\_Networks}.
\end{center}

This paper is organized as follows. In Section 2, we provide some background on phylogenetic models with a particular emphasis on the CFN model and the ideal of invariants for CFN tree models. We also describe the toric fiber product. In Section 3, we show that studying the CFN model on level-1 networks can be reduced to understanding the CFN model on \(n\)-sunlets. In Section 4, we give a complete description for quadratic invariants for any sunlet network. In Section 5, we focus on 4- and 5-leaf sunlet networks and describe some algebraic properties of their ideals. In Section 6, we discuss some open problems and conjectures concerning network ideals and give some possible directions for approaching them. In particular, we conjecture that the CFN sunlet network ideal is generated by quadratics and is dimension $2n$ when the network has $n$ leaves.

\setcounter{tocdepth}{1}

\tableofcontents

\section{Preliminaries}
In this section, we provide some background on phylogenetic networks and phylogenetic Markov models on them. We then discuss toric fiber products which will be useful tools for describing the ideal of phylogenetic invariants for the CFN model on a phylogenetic network. 

\subsection{Phylogenetic Networks}
In this section, we review the basics of phylogenetic networks and define some network structures that we will use throughout the paper. Our notation and terminology is adapted from \cite{GL18, GLR19}. For additional information on the combinatorial properties of networks and definitions associated to them we refer the reader to \cite{GLR19, Ste16}.

\begin{definition}
A phylogenetic network $\cn$ on leaf set $[n] = \{1, 2, \ldots n\}$ is a rooted acyclic digraph with no edges in parallel and satisfying the following properties:
\begin{enumerate}
    \item the root has out-degree two;
    \item a vertex with out-degree zero has in-degree one, and the set of vertices with out-degree zero is $[n]$;
    \item all other vertices have either in-degree one and out-degree two, or in-degree two and out-degree one. 
\end{enumerate}
\end{definition}

Vertices with in-degree one and out-degree two are called \emph{tree vertices} while vertices with in-degree two and out-degree one are called \emph{reticulation vertices}. Edges directed into a reticulation vertex are called \emph{reticulation edges} and all other edges are called \emph{tree edges}. This paper focuses on the CFN model which is group-based and hence \emph{time-reversible}. This means that it is impossible to identify the location of the root under this model so we are only interested in the underlying \emph{semi-directed} network structure of the phylogenetic network. The underlying semi-directed network of a phylogenetic network is obtained by suppressing the root and undirecting all tree edges in the network. The reticulation edges remain directed though. This is illustrated in Figure \ref{fig:rootedNet}. 

As the number of reticulation vertices in the network increases, the parameterization of the model becomes increasingly complicated. A common restriction is to limit the number of reticulation vertices in each biconnected component of the network. A network is called level-$k$ if there is a maximum of $k$ reticulation vertices in each biconnected component of the network. In this paper we will focus on level-1 networks and a special subclass of these networks called \emph{sunlet networks} which were first studied in \cite{GL18}. 

\begin{definition}
A $n$-sunlet network is a semi-directed network with one reticulation vertex and whose underlying graph is obtained by adding a leaf to every vertex of a $n$-cycle. We denote with $\cs_n$ the $n$-sunlet network with reticulation vertex adjacent to the leaf 1 and the other leaves labelled clockwise from 1 in increasing order.
\end{definition}

Note that any level-1 network can be constructed by gluing sunlets of possibly different sizes along trees. It was noted in \cite{GL18} that this corresponds to a toric fiber product of their ideals. We develop this further in Section \ref{sec:sunletTFP}. We end this section with an example that corresponds to the 4-sunlet, $\cs_4$, which we will use throughout this paper.  

\begin{example}
Consider the network pictured on the left in Figure \ref{fig:rootedNet}. This is a 4 leaf, level-1 network. The reticulation edges are dashed and the reticulation vertex is the vertex adjacent to the leaf labelled $1$. It's underlying semi-directed network is pictured on the right. This semi-directed network is a 4-sunlet with reticulation vertex $1$. Observe that deleting either of the reticulation edges in the sunlet network yields an unrooted binary tree with 4 leaves but that these two trees are not the same. 
\end{example}

\begin{figure}
    \centering
    \begin{subfigure}[b]{0.3\linewidth}
        \centering
        \begin{tikzpicture}[scale = .3, thick]
        \draw [fill] (4,8) circle [radius = 0.1]; 
        \draw [fill] (2,6) circle [radius = 0.1];
        \draw [fill] (0,4) circle [radius = 0.1];
        \draw [fill] (4,4) circle [radius = 0.1];
        \draw [fill] (2,2) circle [radius = 0.1];
        \draw [fill] (0,0) circle [radius = 0.1];
        \draw [fill] (-4,0) circle [radius = 0.1];
        \draw [fill] (8,0) circle [radius = 0.1];
        \draw [fill] (12,0) circle [radius = 0.1];
        
        \draw (4,8)--(2,6);
        \draw (2,6)--(0,4);
        \draw (2,6)--(4,4);
        \draw [dashed] (0,4)--(2,2);
        \draw [dashed] (4,4)--(2,2);
        
        \draw (2,2)--(0,0);
        \draw (0,4)--(-4,0);
        \draw (4,4)--(8,0);
        \draw (4,8)--(12,0);
        
        \draw (-4,0) node[below]{$2$};
        \draw (0,0) node[below]{$1$};
        \draw (8,0) node[below]{$4$};
        \draw (12,0) node[below]{$3$};
    \end{tikzpicture}
    \end{subfigure}
    \begin{subfigure}[b]{0.3\linewidth}
        \centering
        \begin{tikzpicture}[scale = .5, thick]
        \draw [dashed] (2,2)--(4,2);
        \draw (4,2)--(4,4);
        \draw (4,4)--(2,4);
        \draw [dashed] (2,4)--(2,2);
        
        \draw (2,2)--(1,1);
        \draw (4,2)--(5,1);
        \draw (4,4)--(5,5);
        \draw (2,4)--(1,5);
        
        \draw (1,1) node[below]{$1$};
        \draw (5,1) node[below]{$2$};
        \draw (5,5) node[above]{$4$};
        \draw (1,5) node[above]{$3$};
        \end{tikzpicture}
    \end{subfigure}
    \caption{A four leaf, level-1 network pictured on the left with all edges directed away from the root. On the right is the associated semidirected network obtained by suppressing the root and undirecting all tree edges. The edges are implicitly assumed to be directed into the vertex adjacent to the leaf 1.}
    \label{fig:rootedNet}
\end{figure}
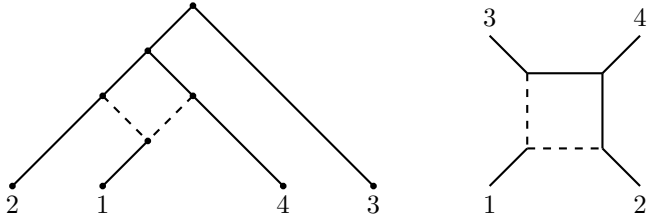

\subsection{Phylogenetic Markov Models}
In this section, we review the basics of phylogenetic Markov models for trees and networks. For additional information we refer the reader to \cite{SS+03, Ste16}. Phylogenetic Markov models on networks are determined by the trees that result from deleting reticulation edges in the network. This means we first need to describe phylogenetic Markov models on trees. 

A $\kappa$-state phylogenetic Markov model on a $n$-leaf, leaf-labelled rooted binary tree $\ct$ gives us a joint distribution on the states of the leaves of $\ct$. This joint distribution is determined by associating a random variable $X_v$ with state space $[\kappa]$ to each internal vertex $v$ of $\ct$ and a $\kappa \times \kappa$ transition matrix $M^e$ to each directed edge $e = (u,v)$ of $\ct$ such that $M_{i,j}^e = P(X_v = j | X_u = i)$. Also associate a root distribution $\pi$ is  to the root $\rho$ of $\ct$. Let $X_i$ be the random variable associated to the leaf labelled $i$ for $i \in [n]$. Then the probability of observing a configuration 
$(x_1, \ldots x_n) \in [\kappa]^n$ of states at the leaves is 
\[
P(X_1 = x_1, \ldots, X_n = x_n) ~= 
\sum_{j \in [\kappa]^{Int(\ct)}}\pi_{j_\rho}\prod_{(u,v) \in E(\ct)}M_{j_u, j_v}^{(u,v)}. 
\]

Note that the joint distribution of $(X_1, \ldots X_n)$ is given by polynomials in the entries of $\pi$ and the $M^e$. This means that the model can be thought of as a polynomial map
\[
\psi_\ct : \Theta_\ct \to \Delta_{\kappa^n-1}
\]
where $\Theta_\ct$ is the \emph{stochastic parameter space} of the model (the space of transition matrices and root distributions) and $\Delta_{\kappa^n-1}$ is the probability simplex. Since this map is a polynomial map, tools from algebraic geometry can be used to study the model. This is one of the key takeaways from algebraic statistics and we refer the reader to \cite{Sul18} for additional information. 

We ignore the restrictions of the stochastic parameter space and extend $\psi_\ct$ to be a complex polynomial map and study the variety $V_\ct = \overline{\mathrm{im}(\psi_\ct)}^{\mathrm{Zar}}$ which is called the \emph{phylogenetic variety} associated to $\ct$. Polynomials in the vanishing ideal $I_\ct = \II(V_\ct)$ are called \emph{phylogenetic invariants} and a major problem for any phylogenetic model is to describe this ideal. Characterizing the invariants of phylogenetic models began with \cite{CF87, L87} and has been continued by many including but not limited to \cite{AR08, CS21, DK09, LS15-Strand, SS05}. 

We can now use the Markov models we have for trees to define phylogenetic Markov models on networks. Let $\cn$ be a network with reticulation vertices $v_1, \dots v_m$ and let $e_i^0$ and $e_i^1$ be the reticulation edges adjacent to $v_i$. Associate a transition matrix to each edge of $\cn$. Independently at random we delete $e_i^0$ with probability $\lambda_i$ and otherwise delete $e_i^1$ and record which edge is deleted with a vector $\sigma \in \{0,1\}^m$ where $\sigma_i = 0$ indicates that edge $e_i^0$ was deleted. Each $\sigma$ corresponds to a different tree $\ct_\sigma$. Then the parameterization $\psi_\cn$ is given by
\begin{equation}
\label{eq:networkParam}
    \psi_\cn = \sum_{\sigma \in \{0,1\}^m}\left(\prod_{i=1}^m \lambda_i^{1 - \sigma_i}(1-\lambda_i)^{\sigma_i} \right)\psi_{\ct_\sigma}
\end{equation}

where $\psi_{\ct_\sigma}$ is the parameterization corresponding to the tree $\ct_\sigma$ with transition matrices inherited from the original network $\cn$. Note that this is similar to a mixture model but with many additional relations among the parameters. The parameterization $\psi_\cn$ is still a polynomial map though which means we can still consider the Zariski closure of the image $\psi_\cn$ and the corresponding ideal of phylogenetic invariants, $I_\cn$. As mentioned previously, if the phylogenetic model is time-reversible then we get the same model by considering the Markov process on the underlying semi-directed network. We end this section with our running example.

\begin{example}
Consider the 4-sunlet $\cs_4$ pictured in figure \ref{fig:4SunletAndTrees} with reticulation vertex adjacent to the leaf 1 and reticulation edges $e_5$ and $e_8$. The trees $\ct_0$ and $\ct_1$ are obtained by deleting edges $e_8$ and $e_5$ respectively. Since there is only one reticulation vertex in $\cs_4$, the sum in Equation \ref{eq:networkParam} simplifies to
\[
\psi_{\cs_4} = \lambda \psi_{\ct_0} + (1-\lambda)\psi_{\ct_1}. 
\]
The transition matrices used in the parameterization maps $\psi_{\ct_i}$ are inherited from the original network. For instance the edge $e_6$ in the original network has a transition matrix $M^{e_6}$ associated to it and thus the edge $e_6$ that appears in $\ct_0$ and the edge $e_6$ that appears in $\ct_1$ both use the same transition matrix $M^{e_6}$.
\end{example}

\begin{figure}
    \centering
    \begin{subfigure}[b]{0.3\linewidth}
        \centering
        \begin{tikzpicture}[scale = .5, thick]
        \draw [dashed] (2,2)--(4,2);
        \draw (4,2)--(4,4);
        \draw (4,4)--(2,4);
        \draw [dashed] (2,4)--(2,2);
        
        \draw (2,2)--(1,1);
        \draw (4,2)--(5,1);
        \draw (4,4)--(5,5);
        \draw (2,4)--(1,5);
        
        \draw (1,1) node[below]{$1$};
        \draw (5,1) node[below]{$4$};
        \draw (5,5) node[above]{$3$};
        \draw (1,5) node[above]{$2$};
        
        \draw (1,1) node[right]{$e_1$};
        \draw (5,1) node[left]{$e_4$};
        \draw (5,5) node[left]{$e_3$};
        \draw (1,5) node[right]{$e_2$};
        
        \draw (3,2) node[below]{$e_8$};
        \draw (4,3) node[right]{$e_7$};
        \draw (3,4) node[above]{$e_6$};
        \draw (2,3) node[left]{$e_5$};
        \end{tikzpicture}
        \caption{$\cs_4$}
    \end{subfigure}
    \begin{subfigure}[b]{0.3\linewidth}
        \begin{tikzpicture}[scale = .5, thick]
        \draw (0,0)--(2,2);
        \draw (2,2)--(4,2);
        \draw (4,2)--(6,0);
        \draw (2,2)--(1,3);
        \draw (4,2)--(5,3);
        
        \draw [fill] (0,0) circle [radius = .1];
        \draw [fill] (1,1) circle [radius = .1];
        \draw [fill] (2,2) circle [radius = .1];
        \draw [fill] (1,3) circle [radius = .1];
        \draw [fill] (4,2) circle [radius = .1];
        \draw [fill] (5,3) circle [radius = .1];
        \draw [fill] (5,1) circle [radius = .1];
        \draw [fill] (6,0) circle [radius = .1];
        
        \draw (0,0) node[below left] {$1$};
        \draw (1,3) node[above left] {$2$};
        \draw (5,3) node[above right] {$3$};
        \draw (6,0) node[below right] {$4$};
        
        \draw (.25,.25) node[right] {$e_1$};
        \draw (1.25,2.75) node[right] {$e_2$};
        \draw (4.75,2.75) node[left] {$e_3$};
        \draw (5.75,.25) node[left] {$e_4$};
        \draw (1.25,1.25) node[right] {$e_5$};
        \draw (3,2) node[below] {$e_6$};
        \draw (4.75,1.25) node[left] {$e_7$};
        \end{tikzpicture}
        \caption{$\ct_0$}
    \end{subfigure}
    \begin{subfigure}[b]{0.3\linewidth}
        \begin{tikzpicture}[scale = .5, thick]
        
        \draw (0,0)--(2,2);
        \draw (2,2)--(4,2);
        \draw (4,2)--(6,0);
        \draw (2,2)--(1,3);
        \draw (4,2)--(5,3);
        
        \draw [fill] (0,0) circle [radius = .1];
        \draw [fill] (1,1) circle [radius = .1];
        \draw [fill] (2,2) circle [radius = .1];
        \draw [fill] (1,3) circle [radius = .1];
        \draw [fill] (4,2) circle [radius = .1];
        \draw [fill] (5,3) circle [radius = .1];
        \draw [fill] (5,1) circle [radius = .1];
        \draw [fill] (6,0) circle [radius = .1];
        
        \draw (0,0) node[below left] {$1$};
        \draw (1,3) node[above left] {$4$};
        \draw (5,3) node[above right] {$2$};
        \draw (6,0) node[below right] {$3$};
        
        \draw (.25,.25) node[right] {$e_1$};
        \draw (1.25,2.75) node[right] {$e_4$};
        \draw (4.75,2.75) node[left] {$e_2$};
        \draw (5.75,.25) node[left] {$e_3$};
        \draw (1.25,1.25) node[right] {$e_8$};
        \draw (3,2) node[below] {$e_7$};
        \draw (4.75,1.25) node[left] {$e_6$};
        \end{tikzpicture}
        \caption{$\ct_1$}
    \end{subfigure}
    \caption{A 4 leaf 4-cycle network $N$ and the two trees $\ct_0$ and $\ct_1$ that are obtained by deleting the reticulation edges $e_8$ and $e_5$ respectively.}
    \label{fig:4SunletAndTrees}
\end{figure}
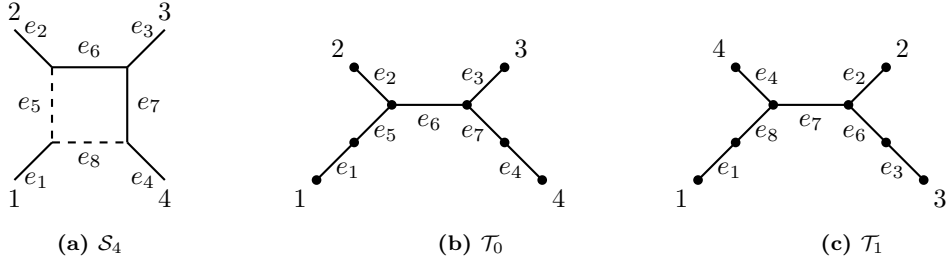

\subsection{The CFN Model}
\label{sec:cfnModel}
In this section, we review the CFN model, sometimes called the binary Jukes-Cantor model, and some known results about the ideal of phylogenetic invariants of trees under this model. In particular, we describe the discrete Fourier transform which turns the map $\psi_\ct$ into a monomial map in the transformed parameters and thus the ideal $I_\ct$ becomes a toric ideal \cite{SS05}. This vastly simplifies the network parameterization as well and will make it much easier to define the parameterization explicitly. We begin with a description of general group-based models and then discuss the CFN model in particular. 

\begin{definition}
Let $G$ be a finite abelian group of order $\kappa$ and $\ct$ a rooted binary tree. The state space of the random variables $X_v$ is identified with the elements of the group $G$. A group-based model on $\ct$ is a phylogenetic Markov model on $\ct$ such that for each transition matrix $M^e$, there exists a function $f_e : G \to \Rr$ such that $M_{g,h}^e = f(g-h)$. 
\end{definition}

The CFN model is a 2 state phylogenetic Markov model where the states are purine (adenine and guanine) and pyrimidine (thymine and cytosine), that is the DNA bases are grouped into two groups corresponding to their chemical structure. It is a group-based model for the group $G = \zmodtwo$ with the states purine and pyrimidine arbitrarily associated to the elements of $\zmodtwo$. This means the transition matrices in the model have the form
\[
M^e = 
\begin{pmatrix}
\alpha & \beta \\
\beta & \alpha
\end{pmatrix}
\]
and the associated function $f_e: \zmodtwo \to \Rr$ is simply $f_e(0) = \alpha$ and $f_e(1) = \beta$. 

Group-based models allow for a linear change of coordinates that makes $\psi_\ct$ a monomial map in the transformed parameters. This means many group-based models (such as the CFN, JC, K2P, and K3P models) are toric varieties in the transformed coordinates \cite{SS05}. This change of coordinates is called the discrete Fourier transform and was first utilized in \cite{ES93, HP96}. The new image coordinates, commonly called the Fourier coordinates, are denoted with $q_{g_1,\ldots, g_n}$ for $g_1,\ldots,g_n \in G$. For the CFN model, that is we have $G = \zmodtwo$, the parameterization $\psi_\ct$ can be given in terms of the edges of the tree and their corresponding splits which we briefly describe first.

A split of $[n]$ is a set partition $A|B$ of the set $[n]$. A split $A|B$ is valid for an unrooted binary tree $\ct$ leaf-labelled by $[n]$ if it can be obtained as the leaf sets of the two connected components of $\ct \setminus e$ for some edge $e$ of $\ct$. So we let $\Sigma(\ct)$ be the set of edges of $\ct$ and to each edge $e$ we associate the split $A_e | B_e$ that deleting the edge $e$ yields. Now for each edge $e \in \Sigma(\ct)$ and each group element $g \in \zmodtwo$ we have a parameter $a_g^e$.
The parameterization of the model $\psi_\ct$ in the Fourier coordinates is given by
\begin{equation}
\label{eqn:treeParam}
q_{g_1,\ldots g_n} = 
\begin{cases}
\displaystyle \prod_{A_e|B_e \in \Sigma(\ct)} a_{\sum\limits_{i \in A_e}g_i}^e & \mbox{ if } \sum\limits_{i \in [n]}g_i = 0 \\ 
0 & \mbox{ otherwise}.
\end{cases}
\end{equation}
Note in the parameterization we are utilizing the natural identification between the edge $e$ and its associated split $A_e|B_e$. 
We can now think of the variety $V_\ct$ as being the closure of the map given in Equation $\ref{eqn:treeParam}$ where the parameters are allowed to range freely over the complex numbers. 

We now introduce two different interpretations of the toric ideal $I_\ct$ that will be useful in building the sunlet ideal $I_{\cs_n}$. Sturmfels and Sullivant showed in \cite{SS05} that the ideal of phylogenetic invariants for a tree $\ct$ under the CFN model can be constructed in the following way. Let $A | B$ be a split of $\ct$ and let $|A| = j$ so $|B| = n - j$. For each $i \in \zmodtwo$ we form a matrix $M_i$ with rows indexed by sequences $\mathbf{r} \in (\zmodtwo)^A$ and columns indexed by sequences $\mathbf{c} \in (\zmodtwo)^B$ such that $\sum_{a \in A} r_a = \sum_{b \in B} c_b = i$. The entry of $M_i$ in row $\mathbf{r}$ and column $\mathbf{c}$ is $q_\bfg$ such that $g|_A = r$ and $g|_B = c$.  Then the ideal of phylogenetic invariants for the tree $\ct$ is given by all of the \(2\times 2\) minors of the matrices $M_i$ as $A|B$ ranges over all the splits of $\ct$. The following example illustrates this construction. 

\begin{example}
\label{ex:cfnTreeIdeal}
Let $\ct$ be the unrooted binary tree determined by the split $12 | 34$. Then
\[
M_0 = 
\begin{blockarray}{ccc}
      & 00 & 11 \\
      \begin{block}{c(cc)}
      00 & q_{0000} & q_{0011}\\
      11 & q_{1100} & q_{1111}\\
      \end{block}
\end{blockarray}
 ~\mathrm{and}~
M_1 = 
\begin{blockarray}{ccc}
      & 01 & 10 \\
      \begin{block}{c(cc)}
      01 & q_{0101} & q_{0110}\\
      10 & q_{1001} & q_{1010}\\
      \end{block}
\end{blockarray}.
\]
So the ideal of phylogenetic invariants for $\ct$ is $I_\ct = \langle q_{0000}q_{1111}-q_{0011}q_{1100},~ q_{0101}q_{1010} - q_{0110}q_{1001} \rangle$. 
\end{example}

Essentially, their construction shows that the ideal $I_\ct$ is given by rank constraints on matrices that come from slicing and flattening the tensor $(q_\bfg : \bfg \in (\zmodtwo)^n)$ according to the splits of $\ct$. This determinantal representation is also amenable to computation since determining whether or not a point is in the variety can be done by verifying that the rank of the associated matrices, $M_i$, is at most one. Another representation of relations in $I_\ct$ was given by Buczy\'{n}ska and Wi\'{s}niewski in \cite{BW07}. They use systems of paths on the tree $\ct$ to describe these binomials instead. Note that any $\bfg \in (\zmodtwo)^n$ defines a unique system of disjoint paths on $T$ that connects the leaves $\ell$ such that $g_\ell = 1$ \cite[Lemma 2.4]{BW07}. One can also construct this path system by including every edge $e$ such that for the associated split $A_e | B_e$ it holds that $\sum_{a \in A_e} g_a = \sum_{b \in B_e} g_b = 1$. The following example illustrates their construction. 

\begin{example}
\label{ex:cfnPathRels}
Let $T$ again be the 4 leaf tree defined by the single split $12|34$. Note that each $\bfg \in (\zmodtwo)^4$ corresponds to a unique system of disjoint paths between the leaves $\ell \in [4]$ such that $g_\ell = 1$. For instance $q_{0101}$ corresponds to the red path
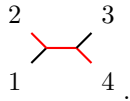
\[
\begin{tikzpicture}[scale = .2, thick]
    \draw (2,2)--(1,1);
    \draw [red] (2,2)--(4,2);
    \draw [red] (4,2)--(5,1);
    \draw [red] (2,2)--(1,3);
    \draw (4,2)--(5,3);
    
    \draw (1,1) node[below left] {$1$};
    \draw (1,3) node[above left] {$2$};
    \draw (5,3) node[above right] {$3$};
    \draw (5,1) node[below right] {$4$};
\end{tikzpicture}.
\]
We saw in Example that $q_{0000}q_{1111}-q_{0011}q_{1100} \in I_T$. Using the interpretation of the variables as paths, we can see this relation as encoding that two systems of paths are equivalent. The paths are pictured below in red.
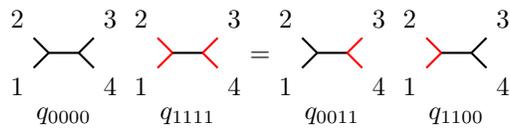
\[
\begin{tikzpicture}[scale = .2, thick]
    \draw (2,2)--(1,1);
    \draw (2,2)--(4,2);
    \draw (4,2)--(5,1);
    \draw (2,2)--(1,3);
    \draw (4,2)--(5,3);
    
    \draw (1,1) node[below left] {$1$};
    \draw (1,3) node[above left] {$2$};
    \draw (5,3) node[above right] {$3$};
    \draw (5,1) node[below right] {$4$};
    \draw (3,-1) node[below] {$q_{0000}$};
\end{tikzpicture}
\begin{tikzpicture}[scale = .2, thick]
    \draw [red] (2,2)--(1,1);
    \draw (2,2)--(4,2);
    \draw [red] (4,2)--(5,1);
    \draw [red] (2,2)--(1,3);
    \draw [red] (4,2)--(5,3);
    
    \draw (1,1) node[below left] {$1$};
    \draw (1,3) node[above left] {$2$};
    \draw (5,3) node[above right] {$3$};
    \draw (5,1) node[below right] {$4$};
    \draw (3,-1) node[below] {$q_{1111}$};
\end{tikzpicture}
\raisebox{27pt}{=}
\begin{tikzpicture}[scale = .2, thick]
    \draw (2,2)--(1,1);
    \draw (2,2)--(4,2);
    \draw [red] (4,2)--(5,1);
    \draw (2,2)--(1,3);
    \draw [red] (4,2)--(5,3);
    
    \draw (1,1) node[below left] {$1$};
    \draw (1,3) node[above left] {$2$};
    \draw (5,3) node[above right] {$3$};
    \draw (5,1) node[below right] {$4$};
    \draw (3,-1) node[below] {$q_{0011}$};
\end{tikzpicture}
\begin{tikzpicture}[scale = .2, thick]
    \draw [red] (2,2)--(1,1);
    \draw (2,2)--(4,2);
    \draw (4,2)--(5,1);
    \draw [red] (2,2)--(1,3);
    \draw (4,2)--(5,3);
    
    \draw (1,1) node[below left] {$1$};
    \draw (1,3) node[above left] {$2$};
    \draw (5,3) node[above right] {$3$};
    \draw (5,1) node[below right] {$4$};
    \draw (3,-1) node[below] {$q_{1100}$};
\end{tikzpicture}
\]
\end{example}

 Since the discrete Fourier transform gives a linear change of coordinates it can also be applied to group-based models on phylogenetic network models \cite{GL18}. This means the parameterization of $\cs_n$ is
 \begin{equation}
\label{eqn:sunletParam}
q_{g_1,\ldots g_n} = 
\begin{cases}
\displaystyle \prod_{A_e|B_e \in \Sigma(\ct_0)} a_{\sum\limits_{i \in A_e}g_i}^e + \prod_{A_e|B_e \in \Sigma(\ct_1)} a_{\sum\limits_{i \in A_e}g_i}^e & \mbox{ if } \sum\limits_{i \in [n]}g_i = 0 \\ 
0 & \mbox{ otherwise}.
\end{cases}
\end{equation}

\begin{example}
Let $\cs_n$ be the 4-sunlet pictured in Figure \ref{fig:4SunletAndTrees}. As we saw in the previous example, the trees $T_0$ and $T_1$ that are also pictured in Figure \ref{fig:4SunletAndTrees} are obtained from $\cs_n$ by deleting the reticulation edges $e_8$ and $e_5$ respectively. We denote the Fourier parameter corresponding to the edge $e_i$ and group element $g_j$ by $a_{g_j}^i$. The parameterization $\psi_{\cs_n}$ in the Fourier coordinates is
\[
q_{g_1, g_2, g_3, g_4} = \begin{cases}
a_{g_1}^1 a_{g_2}^2 a_{g_3}^3 a_{g_4}^4 a_{g_1}^5 a_{g_1 + g_2}^6 a_{g_4}^7 +
a_{g_1}^1 a_{g_2}^2 a_{g_3}^3 a_{g_4}^4 a_{g_3}^6 a_{g_1 + g_4}^7 a_{g_1}^8 
& \mbox{ if } \sum_{i \in [4]}g_i = 0 \\ 
0 & \mbox{ otherwise}
\end{cases}
\]
The first term in the above parameterization comes from the parameterization $\psi_{\ct_0}$ in the Fourier coordinates and the second term comes from $\psi_{\ct_1}$.
\end{example}

This new parameterization is easier to work with than the previous parameterization but $I_{\cs_n}$ is still not a toric ideal in the new coordinates. This means the techniques used to analyze the ideal $I_\ct$ can not be directly used to analyze $I_{\cs_n}$. One of our goals in this paper is to develop new techniques to describe the invariants in $I_{\cs_n}$ that are reminiscent of the original constructions for trees. 

\subsection{Toric Fiber Products}

In this section we recall the \emph{toric fiber product} operation on multigraded ideals first defined by Sullivant in \cite{Sul07}. 

We first consider a polynomial ring $\Cc[\bar{x}] := \Cc[x_1,\dotsc,x_n]$ equipped with a grading by elements of a lattice $M$.  This means that there is a linear map $\deg: \Nn^n \to M$, and a direct sum decomposition of $\Cc[\bar{x}]$ as a $\Cc$-vector space into isotypical components indexed by $M$:

\[ \Cc[\bar{x}] = \bigoplus_{m \in M} \Cc[\bar{x}]_m,\]

\noindent
where $\Cc[\bar{x}]_m$ has a basis of the set of monomials $x^\alpha$ where $\deg(\alpha) = m$. The support semigroup $S(\deg) \subseteq M$ is defined to be the set of $m$ such that $\Cc[\bar{x}]_m \neq 0$. It is straightforward to show that $S(\deg)$ is closed under under addition, contains $0 \in M$, and is generated by the set $\{d_1, \ldots, d_n\} \subset M$, where $d_i = \deg({\bf e}_i)$.

A polynomial $f \in \Cc[\bar{x}]$ is said to be $M$-homogeneous if $f \in \Cc[\bar{x}]_m$ for some $m \in S(\deg)$.  Equivalently, $f$ is $M$-homogeneous if and only if each non-zero monomial term $C_\alpha {\bf x}^\alpha$ appearing in $f$ satisfies $\deg(\alpha) = m$.  A polynomial ideal $I \subseteq \Cc[\bar{x}]$ is $M$-homogeneous if it satisfies the following equivalent conditions:

\begin{enumerate}
    \item $I = \langle f_1, \ldots, f_\ell\rangle$ for $M$-homogeneous polynomials $f_i \in \Cc[\bar{x}]_{m_i}$,
    \item $I = \bigoplus_{m \in M} I_m$, where $I_m = I \cap \Cc[\bar{x}]_m$.
\end{enumerate}


Next we consider two $M$-graded polynomial rings $\Cc[\bar{x}]$, $\Cc[\bar{y}]$ with homogeneous ideals $I \subset \Cc[\bar{x}]$ and $J \subset \Cc[\bar{y}]$. Let $\deg_1: \Nn^n \to M$ and $\deg_2: \Nn^m \to M$ be the linear maps corresponding to the $M$-gradings on $\Cc[\bar{x}]$ and $\Cc[\bar{y}]$, respectively.  We make the technical assumption that the set of degrees $\mathcal{A} = \{d_1, \ldots, d_r\} \subset M$ obtained by applying the functions $\deg_1, \deg_2$ to the generators of $\Cc[\bar{x}]$ and $\Cc[\bar{y}]$ form a linearly independent set in $M$, and we assume without loss of generality that $\mathrm{rank}(M) = r$. We also assume that each degree $d_i$ is realized by an element of $\bar{x}$ and an element of $\bar{y}$. These conditions are satisfied by the toric fiber products of cycle networks we consider in this paper, where $M = \Zz^2$ and the degree set can be taken to be $\{(0,1), (1,0)\}$. 

We define $S \subset \Cc[\bar{x}, \bar{y}]$ to be the subalgebra spanned by those monomials ${\bf x}^\alpha{\bf y}^{\beta}$ such that $\deg_1(\alpha) = \deg_2(\beta)$. It is a straightforward consequence of the linear independence assumption on $\mathcal{A}$ that $S$ is generated as an algebra by the monomials ${\bf x}_i{\bf y}_j$ where ${\bf x}_i$ and ${\bf y}_j$ have the same $M$-degree. We let $\Cc[\bar{z}]$ be the polynomial ring on variables ${\bf z}_{ij}$, where $ij$ corresponds to a monomial ${\bf x}_i{\bf y}_j$ with this property. We let $\phi$ denote the composition of the following ring homomorphisms:

\[\Cc[\bar{z}] \to \Cc[\bar{x}, \bar{y}] \to \Cc[\bar{x}, \bar{y}]/\langle I, J \rangle\]
\[{\bf z}_{ij} \to {\bf x}_i{\bf y}_j \to [{\bf x}_i{\bf y}_j]\]
 
The following is the main definition of \cite{Sul07}.

\begin{definition}
Let $I \subset \Cc[\bar{x}]$ and $J \subset \Cc[\bar{y}]$ be $M$-graded ideals as above, then the toric fiber product $I \times_M J \subset \Cc[\bar{z}]$ is defined to be the kernel of $\phi$.
\end{definition}

A key feature of the toric fiber product construction in the linearly independent case we consider here is that a Gr\"obner basis for $I\times_M J$ can be assembled from Gr\"obner bases for $I$ and $J$. Recall that a weight vector $w \in \Qq^n$ defines an initial ideal $in_w(I) \subset \Cc[\bar{x}]$ (see \cite{GBCP}).  In particular, $in_w(I)$ is generated by the initial forms $in_w(f)$ for $f \in I$, where $in_w(f)$ is the polynomial obtained from $f$ by taking only those monomial terms whose monomial power is minimized on the inner product with $w$.  We say that $G \subset I$ is a Gr\"obner basis for $I$ with respect to $w$ if the initial forms $\{in_w(g) \mid g \in G\} \subset in_w(I)$ are a generating set. 

The kernel of the map $\Cc[\bar{z}] \to \Cc[\bar{x}, \bar{y}]$ is a binomial ideal with a distinguished generating set $Quad_M$.  Following the description in \cite[Proposition 2.6]{Sul07}, we suppose ${\bf x}_{i_1}, {\bf x}_{i_2}$, ${\bf y}_{j_1}$, ${\bf y}_{j_2}$ all have the same degree $d$, then we get a relation:

\[z_{i_1,j_1}z_{i_2,j_2} - z_{i_1,j_2}z_{i_2,j_1}.\]

Ranging over all $d \in \mathcal{A}$ we obtain a set of binomial quadratic relations $\mathrm{Quad}_M \subset I \times_M J$. 

Let $w_1 \in \Qq^n$ and $w_2 \in \Qq^m$ be weights for $\Cc[\bar{x}]$ and $\Cc[\bar{y}]$, respectively.  We obtain a weight $\phi^*(w_1, w_2)$ for $\Cc[\bar{z}]$ by setting $\phi^*(w_1, w_2)[{\bf z}_{ij}] = w_1[{\bf x}_i] + w_2[{\bf y}_j]$.  For an $M$-homogeneous polynomial $g \in \Cc[\bar{x}]$, $Lift(g) \subset \Cc[\bar{z}]$ is obtained as follows. Let $g = \sum C_{{\bf a}} {\bf x}_1^{a_1}\cdots {\bf x}_n^{a_n}$, where $\sum a_i \deg_1({\bf x}_i) = u \in M$ is fixed for all monomials with $C_{\bf a} \neq 0$.   Linear independence of $\mathcal{A}$ implies that for each $d_i \in \mathcal{A}$, the total contribution of $d_i$ in each monomial term is independent of ${\bf a}$. Now, for each ${\bf x}_i$ select $\kappa(i) \in [m]$ such that $deg_1({\bf x}_i) = deg_2({\bf y}_{\kappa(i)})$ and $\kappa(i) = \kappa(j)$ when $\deg_1({\bf x}_i)= \deg_1({\bf x}_j)$.  This choice $\kappa$ defines a set of monomial generators ${\bf z}_{1, \kappa(1)}, \ldots, {\bf z}_{n, \kappa(n)} \in \Cc[\bar{z}]$. The $\kappa$-lift of $g$ is the polynomial $g_\kappa = \sum C_{\bf a} {\bf z}_{1, \kappa(1)}^{a_1}\cdots {\bf z}_{n, \kappa(n)}^{a_n} \in \Cc[\bar{z}]$.  The set $\mathrm{Lift}(g)$ is then defined to be the set of all such lifts. The lift of an $M$-homogeneous polynomial in $\Cc[\bar{y}]$, and the lift of a set of a polynomials are defined similarly.  The following is \cite[Theorem 2.8]{Sul07}.


\begin{proposition}\label{prop-tfp}
Let $G_1 \subset I$ and $G_2 \subset J$ be Gr\"obner bases with respect $w_1$ and $w_2$ respectively, then $\{\mathrm{Lift}(G_1), \mathrm{Lift}(G_2), \mathrm{Quad}_M\}$ is a Gr\"obner basis with respect to $\phi^*(w_1, w_2)$, and $in_{w_1}(I)\times_M in_{w_2}(J) = in_{\phi^*(w_1, w_2)}(I \times_M J)$.
\end{proposition}

\begin{corollary}\label{cor-torictoricfiber}
If $I$ and $J$ have weights $w_1$ $w_2$, respectively, with Gr\"obner bases with degrees bounded above by $k$, then there is a Gr\"obner basis of $I\times_M J$ with respect to $\phi^*(w_1, w_2)$ of degree greater than $2$ and bounded above by $k$. If the initial ideals $in_{w_1}(I)$, $in_{w_2}(J)$ are toric, then $in_{\phi^*(w_1, w_2)}(I\times_M J))$ is toric. 
\end{corollary}

\begin{proof}
If $in_{w_1}(I)$, $in_{w_2}(J)$ are toric, then Proposition \ref{prop-tfp} implies that $in_{\phi^*(w_1, w_2)}(I\times_M J))$ is the kernel of a map to a domain, and possesses a binomial Gr\"obner basis. 
\end{proof}

The assumption that $I$ and $J$ are $M$-homogeneous ideals implies that their factor rings $A = \Cc[\bar{x}]/I$ and $B = \Cc[\bar{y}]/J$ are $M$-graded as well:

\[A = \bigoplus_{m \in M} A_m \ \ \ \  B = \bigoplus_{m \in M} B_m.\]

The linear independence of the set $\AA \subset M$ implies that the factor ring $\Cc[\bar{z}]/I\times_M J$ is isomorphic to the subalgebra $\bigoplus_{m \in M} A_m \otimes B_m \subset A \otimes_{\Cc} B$.  We let $(A \otimes_{\Cc} B)^{T_M}$ denote this subalgebra. This notation is explained as follows.  The spectrum of the group algebra $\Cc[M]$ is an algebraic torus $T_M$, and the $M$-grading on $A$ and $B$ naturally corresponds to an action by $T_M$, where the graded components $A_m$ and $B_m$ are the $m \in M$-isotypical spaces of $A$ and $B$, respectively, when these rings are regarded as $T_M$ representations.  Consequently, we can define an ``anti-diagonal" $T_M$-action on the tensor product $A \otimes_{\Cc} B$ by giving $B_m$ isotypical degree $-m$. The subring $(A \otimes_{\Cc} B)^{T_M} \subset A \otimes_{\Cc} B$ is the ring of invariants with respect to the antidiagonal action. In the following we use the invariant-theoretic interpretation of the toric fiber product.

\begin{proposition}\label{prop-toricfibernormal}
With $I$, $J$, and $\AA$, $A$, and $B$ as above, if $A$ and $B$ are normal, then $(A \otimes_{\Cc} B)^{T_M}$ is normal.  If there exist $w_1$ and $w_2$ such that $\Cc[\bar{x}]/in_{w_1}(I)$ and $\Cc[\bar{y}]/in_{w_2}(J)$ are normal toric algebras, then $(A \otimes_{\Cc} B)^{T_M}$ is normal and Cohen-Macaulay. 
\end{proposition}

\begin{proof}
The invariant ring of a normal algebra is normal, and the invariant ring of a Cohen-Macaulay ring is Cohen-Macaulay. If $\Cc[\bar{x}]/in_{w_1}(I)$ and $\Cc[\bar{y}]/in_{w_2}(J)$ are normal toric algebras, then both are normal and Cohen-Macaulay.  It follows that the algebras $A$ and $B$ are normal and Cohen-Macaulay, and also that $\Cc[\bar{z}]/in_{\phi^*(w_1, w_2)}(I \times_M J)$ is normal and Cohen-Macaulay. We conclude that $(A \otimes B)^{T_M}$ is normal and Cohen-Macaulay as well. 
\end{proof}

We recall the characterization of Gorenstein normal toric algebras \cite[Corollary 6.3.8]{BrunsHerzog}.  Let $P \subseteq \Rr^n$ be a polyhedral cone with affine semigroup $S_P = P \cap \Zz^n$ and relative interior $\interior(P)$.  For $w \in P$ let $[w]$ denote the associated element of the affine semigroup algebra $\Kk[S_P]$, The \emph{canonical module} of $\Kk[S_P]$ is isomorphic to the ideal $\langle [w] \mid w \in \interior(P) \cap \Zz^n \rangle = \Omega_P \subset \Kk[S_P]$.  The algebra $\Kk[S_P]$ is Gorenstein if and only if $\Omega_P = [w]\Kk[S_P]$ for some $w \in \interior(P)$.

\begin{proposition}\label{prop-toricfibergor}
Let $I$, $J$, and $\AA$, $A$, and $B$ be as above.  Suppose there exist $w_1$ and $w_2$ such that $\Kk[\bar{x}]/in_{w_1}(I)$ and $\Kk[\bar{y}]/in_{w_2}(J)$ are Gorenstein normal toric algebras isomorphic to $\Kk[S_P]$ and $\Kk[S_Q]$ for cones $P \subset \Rr^n$ and $Q \subseteq \Rr^m$, respectively. Finally, suppose that the canonical module generators $u \in P$ and $v \in Q$ have the same $M$-degree.  Then $\Kk[\bar{z}]/in_{\phi^*(w_1, w_2)}(I \times_M J)$ and $(A \otimes B)^{T_M}$ are normal Gorenstein algebras. 
\end{proposition}

\begin{proof}
If $(p, q) \in S_P\times_{M} S_Q$ is not of the form $(p', q') + (u, v)$ for $p' \in S_P$ and $q' \in S_Q$, then it follows that $p \in S_P$ or $q \in S_Q$ is not an interior point of $P$ or $Q$, respectively. Say $p$ is not an interior point of $P$.  It follows that for some linear function $\ell: \Rr^n \to \Rr$, $\ell(p) = 0$ and $\ell(w) > 0$.  We extend $\ell$ to $\ell': \Rr^n\times\Rr^m \to \Rr$.  It follows that $\ell'(p, q) = 0$ and $\ell'(u, v) > 0$, so that $(p, q)$ must be on the relative boundary of $P \times_{M} Q$.  By contrapositive, if $(p, q)$ is a relative interior point of $P \times_{M} Q$, then $(p, q) = (p',q') + (u, v)$, for some $(p', q') \in S_P\times_{M} S_Q$. This implies that $\Kk[\bar{z}]/in_{\phi^*(w_1, w_2)}(I \times_M J)$ is normal and Gorenstein.  It follows that $(A\otimes B)^{T_M}$ is normal and Cohen-Macaulay, with the same Hilbert function as its initial algebra $\Kk[\bar{z}]/in_{\phi^*(w_1, w_2)}(I \times_M J)$.  Now \cite[Theorem 4.4]{Stanley} implies that $(A\otimes B)^{T_M}$ is Gorenstein.     
\end{proof}

\section{Reduction to Sunlet Networks}
\label{sec:sunletTFP}
In this section, we show that gluing level-1 networks together along a leaf corresponds to a toric fiber product of their corresponding ideals. This was pointed out in \cite{GL18} but the authors do not prove it. We include a more detailed discussion and the proof here for completeness. This means that the ideal of invariants for any network can be constructed by taking toric fiber products of sunlet networks and trees.

Let $\cn$ be a level-1 network and observe that we can either find an edge $e$ such that when $e$ is cut, $\cn$ is split into two new networks $\cn_-$ and $\cn_+$ where $\cn_-$ and $\cn_+$ are level-1 networks with fewer leaves or that no such $e$ exists in which case $\cn$ is a sunlet network or 3-leaf tree. We can of course recover the network $\cn$ by gluing $\cn_-$ and $\cn_+$ along the edge $e$ which is a leaf of both new networks. We denote the operation of gluing these networks along a leaf edge as $\cn = \cn_- \ast \cn_+$. This operation is pictured in Figure \ref{fig:netTFP}.

We now assume $\cn$ does admit a decomposition $\cn = \cn_- \ast \cn_+$ and denote the ambient polynomial rings of these networks with  $\Cc[q]$, $\Cc[q]_-$, $\Cc[q]_+$. Note that their corresponding ideals $I_{\cn}, I_{\cn_-}, I_{\cn_+}$ are all homogeneous in the grading determined by $\deg(q_\bfg) = e_{g_e}$ where $e_{g_e}$ is the corresponding standard basis vector.

\begin{example}
Let $\cn_-$ be the corresponding network pictured in Figure \ref{fig:netTFP} then
\[
\Cc[q]_- = \Cc[q_\bfg | \bfg = (g_1,g_2,g_3,g_e) \in \Zz_2 \mbox{ and } g_1 + g_2 + g_3 + g_e = 0]
\]
and one can compute explicitly that
\[
I_{N_-} = \langle q_{0000}q_{1111} -  q_{0011}q_{1100} + q_{0101}q_{1010} - q_{0110}q_{1001}\rangle \subseteq \Kk[q_-]. 
\]
We can clearly see that this polynomial is homogeneous of degree $e_0 + e_1 = \begin{pmatrix} 1 \\ 1\end{pmatrix}$ by simply examining the last entry of the label sequence of each monomial. 
\end{example}

\begin{proposition}
Assume $\cn$ is not a sunlet network or 3-leaf tree and let $\cn = \cn_- \ast \cn_+$ be a decomposition of $\cn$ into two smaller level-1 networks. Let each variable $q_\bfg$ in $\Cc[q]$, $\Cc[q]_-$, $\Cc[q]_+$ have degree $e_{g_e}$. 
Then $I_\cn$ is the toric fiber product:
\[
I_\cn = I_{\cn_-} \times_{\mathcal{A}} I_{\cn_+}
\]
with $\mathcal{A} = \{e_0, e_1\}$ linearly independent.  
\end{proposition}
\begin{proof}
We prove this by slightly modifying the parameterization $\psi_\cn$ and then factoring it which is a standard technique introduced in  \cite{Sul07}. Recall that for a tree $\ct$, $I_\ct$ can be thought of as the kernel of the map
\[
\psi_\ct: \Cc[q] \to \Cc[a_g^i : g \in \zmodtwo, ~~  i \in E(\cn)]
\]
given by Equation \ref{eqn:treeParam} and $I_\cn$ is then the kernel of the map
\begin{equation*}
    \psi_\cn = \sum_{\sigma \in \{0,1\}^m}\left(\prod_{i=1}^m \lambda_i^{1 - \sigma_i}(1-\lambda_i)^{\sigma_i} \right)\psi_{\ct_\sigma}.
\end{equation*}
Note that squaring the variables associated to the edge $e$, which are $a_{g_e}^e$, everywhere they appear does not change the parameterization. 
Furthermore, the edge $e$ which we have glued along is an edge in every tree $\ct_\sigma$ and so we can also split each $\ct_\sigma$ along this edge to get two new trees $\ct_\sigma^+$ and $\ct_\sigma^-$. Then we have from \cite[Theorem 3.10]{Sul07} that
\begin{equation}
\label{eqn:treeParamFactors}
    \psi_{\ct_\sigma}(q_\bfg) = \psi_{\ct_\sigma^-}(q_\bfg)\psi_{\ct_\sigma^+}(q_\bfg).
\end{equation}
That is the parameterization for the tree $\ct_\sigma$ factors as a product of the parameterizations for the trees 
$\ct_\sigma^+$ and $\ct_\sigma^-$. 

Without loss of generality let $v_1, \ldots, v_\ell$ be the reticulation vertices of $N$ that lie in $N_-$ and
$v_{l+1}, \ldots, v_m$ be those that lie in $N_+$. Then we can substitute Equation \ref{eqn:treeParamFactors} into $\psi_\cn$ and regroup to get
\begin{align*}
\psi_\cn(q_\bfg) 
&= 
\sum_{\sigma \in \{0,1\}^m}
\left[ \left(\prod_{i=1}^\ell \lambda_i^{1 - \sigma_i}(1-\lambda_i)^{\sigma_i} \right) \psi_{\ct_\sigma^-}(q_\bfg) \right]
\left[\left(\prod_{i=\ell+1}^m \lambda_i^{1 - \sigma_i}(1-\lambda_i)^{\sigma_i} \right) \psi_{\ct_\sigma^-}(q_\bfg) \right] \\
&= 
\left( \sum_{\sigma \in \{0,1\}^\ell}
\left(\prod_{i=1}^\ell \lambda_i^{1 - \sigma_i}(1-\lambda_i)^{\sigma_i} \right) \psi_{\ct_\sigma^-}(q_\bfg) \right)
\left( \sum_{\sigma \in \{0,1\}^{m-\ell}}
\left(\prod_{i=\ell+1}^m \lambda_i^{1 - \sigma_i}(1-\lambda_i)^{\sigma_i} \right) \psi_{\ct_\sigma^+}(q_\bfg) \right) \\
&=
\psi_{\cn_-}(q_\bfg)\psi_{\cn_+}(q_\bfg)
\end{align*}
since trees $T_\sigma^-$ and $T_\sigma^+$ are exactly the trees that appear in the parameterization of $\psi_{\cn_-}$ and $\psi_{\cn_+}$ respectively. 

This implies that $\psi_N$ factors through the map
\begin{align*}
    \phi: \Cc[q] &\to \Cc[q]_- \otimes \Cc[q]_+ \\ 
            q_\bfg  &\mapsto q_{\bfg_-} \otimes  q_{\bfg_+}
\end{align*}
and thus $I_\cn$ is the desired toric fiber product. 
\end{proof}

\begin{remark}
The exact same proof can be used to extend the above proposition to all group-based models on level-1 phylogenetic networks. We present it in terms of the CFN model here since that is the main focus of our paper. 
\end{remark}

The above proposition gives an immediate algorithm for constructing the ideal $I_\cn$ if the ideals for all sunlet networks and trees are known. The original network $\cn$ is recursively decomposed into sunlet networks and trees. One then builds the ideal back up by taking toric fiber products of the sunlet network ideals and tree ideals. Since the ideals corresponding to trees are completely known, the problem of finding the ideal $I_\cn$ now amounts to understanding the sunlet network ideals $I_{\cs_n}$. This is our main focus for the remainder of this paper.  

\begin{figure}
    \centering
    \begin{subfigure}[b]{0.3\linewidth}
        \centering
        \begin{tikzpicture}[scale = .5, thick]
        \draw [dashed] (2,2)--(4,2);
        \draw (4,2)--(4,4);
        \draw (4,4)--(2,4);
        \draw [dashed] (2,4)--(2,2);
        
        \draw (2,2)--(1,1);
        \draw (4,2)--(5,1);
        \draw (4,4)--(5,5);
        \draw (2,4)--(1,5);
        
        \draw (1,1) node[below]{$1$};
        \draw (5,1) node[below]{$e$};
        \draw (5,5) node[above]{$3$};
        \draw (1,5) node[above]{$2$};
        \end{tikzpicture}
        \caption{$N_-$}
    \end{subfigure}
    \begin{subfigure}[b]{0.3\linewidth}
        \centering
        \begin{tikzpicture}[scale = .5, thick]
        \draw [dashed] (2,2)--(4,2);
        \draw (4,2)--(4,4);
        \draw (4,4)--(2,4);
        \draw [dashed] (2,4)--(2,2);
        
        \draw (2,2)--(1,1);
        \draw (4,2)--(5,1);
        \draw (4,4)--(5,5);
        \draw (2,4)--(1,5);
        
        \draw (1,1) node[below]{$4$};
        \draw (5,1) node[below]{$5$};
        \draw (5,5) node[above]{$6$};
        \draw (1,5) node[above]{$e$};
        \end{tikzpicture}
        \caption{$N_+$}
    \end{subfigure}
    \begin{subfigure}[b]{0.3\linewidth}
        \centering
        \begin{tikzpicture}[scale = .5, thick]
        \draw [dashed] (2,2)--(4,2);
        \draw (4,2)--(4,4);
        \draw (4,4)--(2,4);
        \draw [dashed] (2,4)--(2,2);
        
        \draw (2,2)--(1,1);
        \draw (4,2)--(5,1);
        \draw (4,4)--(5,5);
        \draw (2,4)--(1,5);
        
        \draw (5,1)--(7,1);
        \draw (7,1)--(7,-1);
        \draw [dashed] (7,-1)--(5,-1);
        \draw [dashed] (5,-1)--(5,1);
        
        \draw (7,1)--(8,2);
        \draw (7,-1)--(8, -2);
        \draw (5,-1)--(4, -2);
        
        \draw (1,1) node[below]{$1$};
        \draw (4.5,1.5) node[left, below]{$e$};
        \draw (5,5) node[above]{$3$};
        \draw (1,5) node[above]{$2$};
        \draw (8,2) node[above]{$6$};
        \draw (8,-2) node[below]{$5$};
        \draw (4,-2) node[below]{$4$};
        \end{tikzpicture}
        \caption{$N$}
    \end{subfigure}
    \caption{We can glue two four leaf networks along identified leaves to get a six leaf network. This corresponds to taking a toric fiber product of the corresponding ideals.}
    \label{fig:netTFP}
\end{figure}
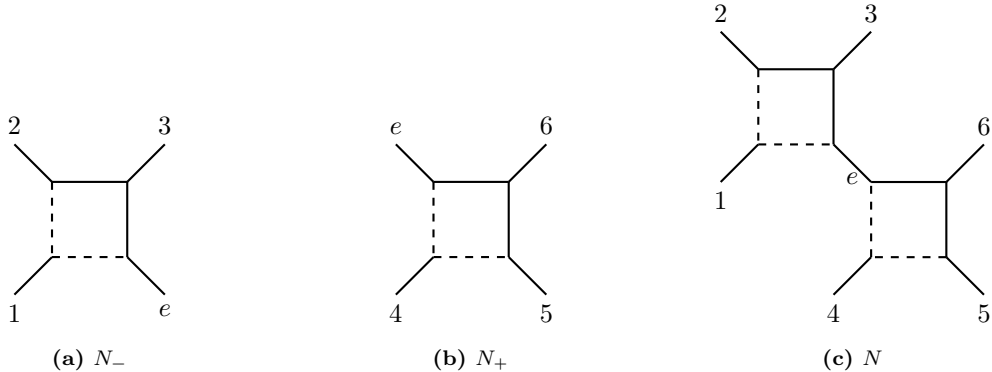

\section{Quadratic Invariants of Sunlet Networks}
\label{sec:gloves}

\subsection{Sunlet Networks are Graded}
\label{sec:Zgrading}

Let \(\cs_n\) be the \(n\)-sunlet network. The leaf edges are labelled \(e_1,\dotsc, e_n\), the reticulation edges are $e_{n+1}$ and $e_{2n}$, and all the other edges are listed sequentially around the cycle clockwise. Let 
\[
R_n = \Cc[q_{g_1,\dotsc,g_n} ~|~ (g_1,\dotsc,g_n) \in (\zmodtwo)^n \text{ and }\sum_{i=1}^n g_i = 0]
\] 
and let 
\[
S_n = \mathbb{C}[a_{g}^i ~|~ g \in \zmodtwo \text{ and }1 \leq i \leq 2n].
\]
Then the defining ideal of \(\cs_n\) is given by the kernel of \(\psi_n : R_n \to S_n\) defined by
\[
q_{g_1,\dotsc, g_n} \mapsto \prod_{j=1}^n a_{g_j}^j \left(\prod_{j=1}^{n-1} a_{\sum_{\ell =1}^{j} g_\ell}^{n+j}  + \prod_{j=2}^{n} a_{\sum_{\ell =2}^{j} g_\ell}^{n+j} \right).
\]
We grade \(R_n\) by \(\Zz^{n+1}\) as follows:
\[
\deg(q_{g_1,\dotsc,g_k}) = (1, h_1,\dotsc, h_k)
\]
where \(h_i = 1\) if \(g_i = \overline{1}\) and \(h_i = 0\) if \(g_i = \overline{0}\). Similarly, we grade \(S_n\) by \(\Zz^{n+1}\) as follows:
\[
    \deg(a_{g_j}^j) = \begin{cases}
        {\bf 0} & \text{if } j > n \\
        (1,h_1, 0, \dotsc, 0) &\text{if } j = 1 \\
        (0,0,\dotsc, h_j, \dotsc, 0) & \text{if } 2\leq j \leq n
    \end{cases}
\]
where the \(h_j\)'s are defined as above and it occurs in the \((j+1)^\text{st}\) position. In this way, we see that \(\psi_n\) is \(\Zz^{n+1}\)-graded \(\Cc\)-algebra homomoprhism; thus, the kernel of \(\psi_n\) inherits the grading on \(R_n\).

\begin{remark}
We have shown that the coordinate ring of the \(n\)-sunlet variety is graded by \(\Zz^{n+1}\). In particular, this makes the variety into a \(T\)-variety: there is a \(T \isom (\Cc^\times)^{n+1}\)-action on the variety. We note that \(\cs_n\) does {\it not} yield a toric variety since in general \(\dim(T) < \dim \cs_n\).
\end{remark}

\subsection{Quadratic phylogenetic invariants for sunlet networks}
\label{sec:NetworkGloves}

In this subsection, we will leverage the grading from Section \ref{sec:Zgrading} to find all quadratic invariants of \(\cs_n\). At first glance, this procedure might feel slightly unnatural; however, as we will see in Section \ref{sec:TreeGloves}, our approach produces all quadratic invariants for any phylogenetic tree. Since these ideals are generated by quadratics this completely describes all invariants for trees, and so we argue that this is a natural procedure to try on networks. At the end of this subsection, we will also give a visual representation of the quadratic invariants in terms of paths in the network.

Throughout this section, let \(\psi_n : R_n \to S_n\) be the parameterization of the network variety \(\cs_n\) as defined in Section \ref{sec:Zgrading}, and let \(J_n = \ker \psi_n\). We begin with a definition.

\begin{definition}
\label{gloves}
Fix $\FF \subseteq [n]$ and $\bfa \in (\zmodtwo)^\FF$. 
The \emph{glove}, \( \cg(n, \FF, \bfa) \), is the \(\Cc\)-vector space spanned by all quadratic monomials \(q_\bfg q_\bfh\) in \(R_n\) so that \({\bf g}\vert_\mathcal{F} = {\bf h}\vert_\mathcal{F} = {\bf a}\) and \({\bf g}\vert_{\mathcal{F}^c} + {\bf h}\vert_{\mathcal{F}^c} = {\bf 1}\) where \({\bf 1}\) is the all ones vector in \((\zmodtwo)^{\FF^c}\). If \(\FF = \emptyset\), then we simply write \(\cg(n,\emptyset)\).
\end{definition}

\begin{remark}
It is not efficient to consider all possible gloves since for some choices of \(\FF\) and \(\bfa\), the corresponding glove intersects \(J_n\) trivially. In fact, given a glove, \(\cg(n,\FF,\bfa) \subseteq R_n\), if \(\cg(n,\FF,\bfa) \cap J_n \neq \{0\}\), then \(|[n] \setminus \FF| \geq 4\) and is even. In order to prove the claim, we first show that if \(|[n]\setminus\FF|\) is odd, then \(\cg(n,\FF,\bfa) = \{0\}\). Indeed, if one considers a monomial \(q_\bfg q_\bfh \in \cg(n,\FF,\bfa)\), then it is not possible for \(\sum_{i=1}^n g_i\) and \(\sum_{i=1}^n h_i\) to both be 0; hence, no such monomial exists. Now, suppose that \(|[n]\setminus \FF| = 0\) or \(2\). In either case, \(\dim_\Cc(\cg(n,\FF,\bfa)) = 1\). Then as \(\psi_n(q_\bfg) \neq 0\) for any \(\bfg\) and as \(S_n\) is an integral domain, all non-trivial polynomials from \(\cg(n,\FF,\bfa)\) lie outside the kernel of \(\psi_n\).
\end{remark}

\begin{remark}
Note that when \(n \geq 4\) and is even, \(\dim_\Cc(\cg(n,\emptyset)) = 2^{n - 2}\). One way to see this is to note that the indices \(\{\bfg,\bfh\}\) that appear in the chosen basis for \(\cg(n,\emptyset)\) are exactly the cosets of \(\langle {\bf 1} \rangle \leq \{\bfg \in (\zmodtwo)^n ~|~ \sum_{i=1}^n g_i = 0\}\).
\end{remark}

With respect to the \(\Zz^{n+1}\) grading from Section \ref{sec:Zgrading}, each glove \(\cg(n,\FF,\bfa)\) is \((R_n)_\bfc\) where \(c_1 = 2\) and \(c_{i+1} = 1\) if \(i \notin \FF\) and \(c_{i+1} = 2a_i\) when \(i \in \FF\). Moreover, this encompasses all graded components whose total degree is 2. Therefore, in order to describe all quadratic phylogenetic invariants of \(\cs_n\), it is enough to find a basis for \(\cg(n,\FF,\bfa) \cap J_n\) for each choice of \(\FF\) and \(\bfa\) where \(|[n] \setminus \FF| = 2k\) for all \(k\) in \(\{1,\dotsc, \lfloor \frac{n}{2} \rfloor\}\). 

In order to state the main result of this section, we need to define two linear maps obtained out of a glove \(\cg(n,\FF,\bfa)\). Consider two following subsets of \([n]\):
\begin{align*}
    \mathbb{E}(n,\FF) &= \{i ~|~ |[i] \setminus \FF| \text{ is even and }2 \leq i \leq n-1\}\\
    \mathbb{O}(n, \FF) &= \{i ~|~ |[i] \setminus \FF| \text{ is odd and }2 \leq i \leq n-1\}.
\end{align*}
When \(n\) and \(\FF\) are clear from context, we will just write \(\mathbb{E}\) and \(\mathbb{O}\), respectively. Using these subsets of \(\{2,\dotsc, n-1\}\), we color the monomials lying in \(\cg(n,\FF,\bfa)\) in two ways. If we have a monomial lying in \(\cg(n,\FF,\bfa)\), and we know that one of the factors is \(q_\bfg\), then the other factor is determined by \(\bfg\). Thus, it is convenient for us to only record ``half" of each term, so we set
\[
    L(n,\FF,\bfa) = \{ \bfg ~|~ q_\bfg q_\bfh \in \cg(n,\FF,\bfa) \text{ and } \bfg <_\text{lex} \bfh\}.
\]
If \(q_\bfg q_\bfh \in \cg(n,\FF,\bfa)\) and \(\bfg \in L(n,\FF,\bfa)\), we define our two colorings as follows.
\begin{align*}
    c_\mathbb{E}(q_\bfg q_\bfh) &= \left(\sum_{i=1}^j g_i\right)_{j \in \mathbb{E}} \in (\zmodtwo)^\mathbb{E}\\
    c_\mathbb{O}(q_\bfg q_\bfh) &= \left(\sum_{i=1}^j g_i\right)_{j \in \mathbb{O}} \in (\zmodtwo)^\mathbb{O}
\end{align*}

Now we can define our two maps \(M_\mathbb{E}^{n, \FF,\bfa} : \cg(n,\FF,\bfa) \to \Cc^{(\zmodtwo)^\mathbb{E}}\) and \(M_\mathbb{O}^{n, \FF,\bfa} : \cg(n,\FF,\bfa) \to \Cc^{(\zmodtwo)^\mathbb{O}}\). With respect to the bases \(\{q_\bfg q_{\bfh} ~|~ \bfg \in L(n,\FF,\bfa\}\), \(\{e_\bfc ~|~ \bfc \in (\zmodtwo)^\mathbb{E}\}\), and \(\{e_\bfc ~|~ \bfc \in (\zmodtwo)^\mathbb{O}\}\), these maps have the following matrix representations.
\[
    (M_\mathbb{E}^{n, \FF,\bfa})_{(\bfc, q_\bfg q_\bfh)} = 
        \begin{cases} 
            1 & \text{if }\bfc = c_\mathbb{E}(q_\bfg q_\bfh) \\
            0 & \text{otherwise}
        \end{cases}
\]
and 
\[
    (M_\mathbb{O}^{n, \FF,\bfa})_{(\bfc, q_\bfg q_\bfh)} = 
        \begin{cases} 
            1 & \text{if }\bfc = c_\mathbb{O}(q_\bfg q_\bfh) \\
            0 & \text{otherwise}
        \end{cases}.
\]
At this point, we are fully equipped to state the main theorem of this section; however, we will delay the proof until Section \ref{sec:GlovesProof}.
\begin{theorem}
\label{thm:MainThm}
Let \(\cg(n,\FF,\bfa)\) be a glove so that either \(1\) is not in \(\FF\) or \(1\) is in \(\FF\) but \(a_1 = 1\). Then 
\[
    J_n \cap \cg(n,\FF,\bfa) = \ker M_{\mathbb{E}}^{n,\FF,\bfa} \cap \ker M_\mathbb{O}^{n, \FF,\bfa}. 
\]
On the other hand, if \(1\) is in \(\FF\) and \(a_1 = 0\), then 
\[
    J_n \cap \cg(n,\FF,\bfa) = \ker M_\mathbb{E}^{n,\FF,\bfa}.
\]
\end{theorem}
\begin{remark}
As we shall see in Section \ref{sec:TreeGloves}, Theorem \ref{thm:MainThm} can be reformulated as follows: \(f \in J_n \cap \cg(n,\FF,\bfa)\) if and only if \(f\) is a phylogenetic invariant for both underlying trees. If we let \(I_{\ct_0}\) and \(I_{\ct_1}\) be the defining ideals for the two underlying trees, then it is always true that \(J_n\) is contained in the intersection of \(I_{\ct_0}\) and \(I_{\ct_1}\); however, in general, \(J_n\) is \emph{not} the intersection of these two toric ideals as can be seen even when \(n = 4\). Indeed, the ideals for the two underlying trees are given by 
\[
I_{\ct_0} = \langle {q}_{0011} {q}_{1100}-q_{0000}q_{1111}, q_{0110}q_{1001} - q_{0101}q_{1010} \rangle
\]
\[
I_{\ct_1} = \langle q_{0101}q_{1010} - q_{0011}q_{1100}, q_{0110}q_{1001} - q_{0000}q_{1111} \rangle.
\]
However, \(I_{\ct_0} \cap I_{\ct_1}\) is generated by one quadratic and one quartic, while \(J_4\) is generated by just the quadratic.
\end{remark}

\begin{proposition}
\label{prop:dimensionofgloves}
If \(n\) is at least 4 and is even, then \(\dim_\Cc(J_n \cap \cg(n,\emptyset)) = (2^{n/2 - 1} - 1)^2\). Moreover, as long as \(1 \notin \FF\) or \(1 \in \FF\) but \(a_1 = 1\), \(J_n \cap \cg(n,\FF,\bfa) \isom J_{n - |\FF|} \cap \cg(n,\emptyset)\).
\end{proposition}

\begin{proof}
For the first claim, note that \(n\) must be even; otherwise, \(\cg(n,\emptyset)\) is trivial. By Theorem \ref{thm:MainThm}, \(J_n \cap \cg(n,\emptyset)\) is the intersection of \(\ker M_\mathbb{E}^{n,\emptyset}\) and \(\ker M_\mathbb{O}^{n,\emptyset}\). Let \(M^{n,\emptyset}\) be the map \(M_\mathbb{E}^{n,\emptyset} \oplus M_\mathbb{O}^{n,\emptyset} : \cg(n,\emptyset) \to \Cc^{(\zmodtwo)^\mathbb{E}} \oplus \Cc^{(\zmodtwo)^\mathbb{O}}\), so \(J_n \cap \cg(n,\emptyset) = \ker M^{n,\emptyset}\).

We will demonstrate that \(\dim_\Cc(J_n \cap \cg(n,\emptyset)) = (2^{n/2 - 1} - 1)^2\) by showing that the rank of \(M^{n,\emptyset}\) is \(2^{n/2} - 1\). Then, as \(\dim_\Cc(\cg(n,\emptyset)) = 2^{n-2}\), we will see by rank-nullity that \(\dim_\Cc(J_n \cap \cg(n,\emptyset)) = 2^{n-2} - 2^{n/2} + 1 = (2^{n/2 - 1} - 1)^2\).

Note that \(|\mathbb{E}| = |\mathbb{O}| = \frac{n}{2} - 1\). If we think of \(M^{n,\emptyset}\) as a matrix, its columns are indexed by monomials \(q_\bfg q_\bfh \in \cg(n,\emptyset)\), and its first \(2^{n/2 - 1}\) rows are indexed by the elements of \((\zmodtwo)^\mathbb{E}\) and the last \(2^{n/2 - 1}\) rows are indexed by \((\zmodtwo)^\mathbb{O}\). We claim that the matrix for \(M^{n,\emptyset}\) takes the following form: (1) every column is of the form \(e_{\bfc_1} + e_{\bfc_2}\) where \(\bfc_1 \in (\zmodtwo)^\mathbb{E}\) and \(\bfc_2 \in (\zmodtwo)^\mathbb{O}\), (2) each column is distinct, and (3) every possible combination of \(e_{\bfc_1} + e_{\bfc_2}\) occurs.

The first point is clear by the definition of the maps \(M_\mathbb{E}^{n,\emptyset}\) and \(M_\mathbb{O}^{n,\emptyset}\). For the second and third points, we will show that for any \(\bfc_1 \in (\zmodtwo)^\mathbb{E}\) and \(\bfc_2 \in (\zmodtwo)^\mathbb{O}\) there is a unique \(q_\bfg q_\bfh \in \cg(n,\emptyset)\) so that \(c_\mathbb{E}(q_\bfg q_\bfh) = \bfc_1\) and \(c_\mathbb{O}(q_\bfg q_\bfh) = \bfc_2\). Note that uniqueness will follow immediately since if there were two monomials whose colors are \(\bfc_1\) and \(\bfc_2\), then they must be the same since \(\bfc_1\) and \(\bfc_2\) record all the partial sums of each of their corresponding group elements. We will build up \(\bfg \in L(n,\FF,\bfa)\) whose partial sums are given by \(\bfc_1\) and \(\bfc_2\). Let \(\bfc \in (\zmodtwo)^{\{2,\dotsc, n-1\}}\) be the unique vector with \(\bfc\vert_\mathbb{E} = \bfc_1\) and \(\bfc\vert_\mathbb{O} = \bfc_2\). If we let \(\widetilde{\bfc} = (0,\bfc,0) \in (\zmodtwo)^n\), then we set \(g_i = \widetilde{c}_i + \widetilde{c}_{i-1}\) for \(i \geq 2\) and \(g_1 = 0\). One can see that \(\sum_{i = 1}^j g_j= c_j\) for any \(2 \leq j \leq n-1\). In order to get a monomial in the glove, we consider \(q_\bfg q_{{\bf 1} + \bfg} \in \cg(n,\emptyset)\). By construction, \(c_\mathbb{E}(q_\bfg q_{{\bf 1} + \bfg}) = \bfc_1\) and \(c_\mathbb{O}(q_\bfg q_{{\bf 1} + \bfg}) = \bfc_2\).

Now, we can show that the row rank of \(M^{n,\emptyset}\) is one less than the number of rows. Up to scaling there is only one linear relation among the rows which is given by adding up the first \(2^{n/2 - 1}\) rows and subtracting off the last \(2^{n/2 -1}\) rows. Points (2) and (3) above guarantee that this is the only relation among the rows. Since the rank of \(M^{n,\emptyset}\) is \(2^{n/2} - 1\) and \(\dim_\Cc \cg(n,\emptyset) = 2^{n-2}\), we have that \(\dim_\Cc (J_n \cap \cg(n,\emptyset)) = (2^{n/2 - 1} - 1)^2\).

For the second statement fix a glove \(\cg(n,\FF,\bfa)\). First, suppose that \(\sum_{i \in \FF} a_i = 0\). Then for any \(\bfg \in (\zmodtwo)^{n - |\FF|}\), define \(\bfg(\FF,\bfa) \in (\zmodtwo)^n\)
as \(\bfg(\FF,\bfa)\vert_\FF = \bfa\) and \(\bfg(\FF,\bfa)\vert_{\FF^c} = \bfg\). Then define a linear map \(T : \cg(n,\emptyset) \to \cg(n,\FF,\bfa)\) defined by \(T(q_\bfg q_\bfh) = q_{\bfg(\FF,\bfa)} q_{{\bf 1} + \bfg(\FF,\bfa)}\). \(T\) is an isomorphism, and it is not hard to see that there is a map which makes the diagram commute and is an isomorphism when restricted to the images of the horizontal maps. 
\begin{center}
    \begin{tikzcd}
        \cg(n - |\FF|, \emptyset) \arrow[r, "M^{n - |\FF|, \emptyset}"] \arrow[d, "T"] & \Cc^{\mathbb{E}(n - |\FF|, \emptyset)} \oplus \Cc^{\mathbb{O}(n - |\FF|, \emptyset)}  \arrow{d}\\
        \cg(n, \FF, \bfa) \arrow[r,"M^{n, \FF,\bfa}"] & \Cc^{\mathbb{E}(n,\FF)} \oplus \Cc^{\mathbb{O}(n, \FF)}
    \end{tikzcd}
\end{center}
It then follows that \(J_n \cap \cg(n,\FF,\bfa) \isom J_{n - |\FF|} \cap \cg(n,\emptyset)\) in this case. The other case, when \(\sum_{i \in \FF} a_i = 1\), is exactly the same except \(\bfg(\FF,\bfa)\) is defined as \(\bfg(\FF,\bfa)\vert_\FF = \bfa\) and \(\bfg(\FF,\bfa)\vert_{\FF^c} = \bfg + e_{n - |\FF|}\).
\end{proof}

By the propsoition, in order to find a basis for \(J_n \cap \cg(n,\FF,\bfa)\) when \(1\) is not in \(\FF\) or \(1\) is in \(\FF\) but \(a_1 = 1\), it is enough to find a basis for \(J_{n - |\FF|} \cap \cg(n - |\FF|, \emptyset)\) and then apply the map \(T\). In the next proposition, we provide an explicit basis for \(J_n \cap \cg(n,\emptyset)\) for any even \(n\) greater than or equal to 4.

\begin{theorem}
\label{prop:idealgenerators}
Fix an even integer \(n \in \Zz_{\geq 4}\), and a group element \(\bfc \in (\zmodtwo)^{\{2,\dotsc, n-1\}}\) so that \(\bfc\vert_{\mathbb{E}(n,\emptyset)} \neq {\bf 0}\) and \(\bfc\vert_{\mathbb{O}(n,\emptyset)} \neq {\bf 0}\). Then we define the polynomial 
\[
f_\bfc = q_{\bfg({\bf 0}, {\bf 0})}q_{\bfh({\bf 0}, {\bf 0})} - q_{\bfg(\bfc\vert_\mathbb{E}, {\bf 0})}q_{\bfh(\bfc\vert_\mathbb{E}, {\bf 0})} + q_{\bfg(\bfc\vert_\mathbb{E}, \bfc\vert_\mathbb{O})}q_{\bfh(\bfc\vert_\mathbb{E}, \bfc\vert_\mathbb{O})} - q_{\bfg({\bf 0}, \bfc\vert_\mathbb{O})} q_{\bfh({\bf 0}, \bfc\vert_\mathbb{O})}
\] 
in \(J_n \cap \cg(n,\emptyset)\). Here \(\bfg(\bfc', \bfc'')\) is defined by setting \(g_1 = 0\) and for \(i \geq 2\) we have that \(g_i = c_{i-1} + c_{i}\) where \(\bfc \in (\zmodtwo)^n\) has \(c_1 = c_n = 0\) and \(\bfc\vert_\mathbb{E} = \bfc'\) and \(\bfc\vert_\mathbb{O} = \bfc''\), and \(\bfh(\bfc_1,\bfc_2) = {\bf 1} + \bfg(\bfc_1, \bfc_2)\). Then 
\[
    \mathcal{B}_n = \{ f_\bfc ~|~ \bfc \in (\zmodtwo)^{\{2,\dotsc,n-1\}} \text{ and } \bfc\vert_\mathbb{E} \neq {\bf 0}, \bfc\vert_\mathbb{O} \neq {\bf 0}\}
\] 
is a basis for \(J_n \cap \cg(n,\emptyset)\).
\end{theorem}

\begin{proof}
Note that by definition \(f_\bfc \in \cg(n,\FF,\bfa)\). To see that \(f_\bfc\in J_n\), note that 
\[
    M^{n,\emptyset}( f_\bfc) = e_{{\bf 0}\vert_\mathbb{E}} + e_{{\bf 0}\vert_\mathbb{O}} - e_{{\bfc}\vert_\mathbb{E}} - e_{{\bf 0}\vert_\mathbb{O}} + e_{{\bfc}\vert_\mathbb{E}} + e_{{\bfc}\vert_\mathbb{O}}  - e_{{\bf 0}\vert_\mathbb{E}} - e_{{\bfc}\vert_\mathbb{O}}  = 0.
\]
By Theorem \ref{thm:MainThm}, \(f_\bfc \in J_n\).

Since \(|\mathcal{B}_n|\) is \((2^{n/2 - 1} - 1)^2\), it is enough show that \(\mathcal{B}_n\) is independent. Consider any linear combination of the elements of \(\mathcal{B}_n\)
\[
    0 = \sum_{\bfc} a_\bfc f_\bfc.
\]
Projecting \(\sum_{\bfc} a_\bfc f_\bfc\) onto \(\mathrm{span}_\Cc\{q_{\bfg(\bfc\vert_\mathbb{E}, \bfc\vert_\mathbb{O})}q_{\bfh(\bfc\vert_\mathbb{E}, \bfc\vert_\mathbb{O})}\}\), yields \(a_\bfc q_{\bfg(\bfc\vert_\mathbb{E}, \bfc\vert_\mathbb{O})}q_{\bfh(\bfc\vert_\mathbb{E}, \bfc\vert_\mathbb{O})}\) from which it follows that \(a_\bfc = 0\) for all such \(\bfc\).
\end{proof}

\begin{remark}
Let \(I_n\) be the ideal generated by all quadratics in \(J_n\). Then Propositions \ref{prop:dimensionofgloves} and \ref{prop:idealgenerators} give a recipe for obtaining generators of \(J_n \cap \cg(n,\FF,\bfa)\) where either \(1 \notin F\) or \(1 \in \FF\) but \(a_1 = 0\). The case when \(1 \in \FF\) and \(a_1 = 0\) is easily taken care of using previously known technology. In this case, the parameterization restricts to a monomial map. These phylogenetic invariants are obtained from the underlying tree \(\ct\) in \(\cs_n\) where all the edges containing the reticulation vertex are deleted. Of course, this tree only has \(n-1\) leaves, so we lift these phylogenetic invariants to the network via the map \(q_\bfg \mapsto q_{(0,\bfg)}\). These facts along with Propositions 4.7 and 4.8 allow us to find all quadratic generators of the sunlet network ideal very quickly. Our implementation of this can be found in the macaulay2 file \texttt{sunletQuadGens.m2}.
\end{remark}

Similar to the tree case, each variable $q_\bfg$ can be thought of as a system of paths on the network. The paths connecting the vertices $\ell$ such that $g_\ell = 1$ though are not necessarily unique. Indeed, there is a unique system of paths connecting all such vertices for each tree. For a monomial, \(q_\bfg\), we consider all the edges in the network which are supported in either of these two path systems. Now, we fix a glove \(\cg(n,\FF,\bfa)\) so that \(1 \notin \FF\). For any monomial \(q_\bfg q_\bfh \in \cg(n,\FF,\bfa)\), we take the symmetric difference of the collection of edges obtained from each monomial. Below is an example with \(q_{001100}q_{100010} \in \cg(\{2,6\}, (0,0)) \subset R_6\).
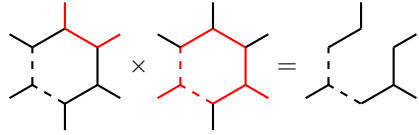
\[{
    \begin{tikzpicture}[scale = 0.5, thick]
        \draw [red] ({sin(0)}, {cos(0)})--({sin(60)}, {cos(60)});
        \draw  ({sin(60)}, {cos(60)})--({sin(120)}, {cos(120)});
        \draw  ({sin(120)}, {cos(120)})--({sin(180)}, {cos(180)});
        \draw [dashed] ({sin(180)}, {cos(180)})--({sin(240)}, {cos(240)});
        \draw [dashed] ({sin(240)}, {cos(240)})--({sin(300)}, {cos(300)});
        \draw ({sin(0)}, {cos(0)})--({sin(300)}, {cos(300)});
        
        \draw [red] ({sin(0)}, {cos(0)})--({1.7*sin(0)}, {1.7*cos(0)});
        \draw [red] ({sin(60)}, {cos(60)})--({1.7*sin(60)}, {1.7*cos(60)});
        \draw ({sin(120)}, {cos(120)})--({1.7*sin(120)}, {1.7*cos(120)});
        \draw ({sin(180)}, {cos(180)})--({1.7*sin(180)}, {1.7*cos(180)});
        \draw ({sin(240)}, {cos(240)})--({1.7*sin(240)}, {1.7*cos(240)});
        \draw ({sin(300)}, {cos(300)})--({1.7*sin(300)}, {1.7*cos(300)});
        
    \end{tikzpicture}
    \;\raisebox{22pt}{\(\times\)}\;
    \begin{tikzpicture}[scale = 0.5, thick]
        \draw [red] ({sin(0)}, {cos(0)})--({sin(60)}, {cos(60)});
        \draw [red] ({sin(60)}, {cos(60)})--({sin(120)}, {cos(120)});
        \draw [red] ({sin(120)}, {cos(120)})--({sin(180)}, {cos(180)});
        \draw [red, dashed] ({sin(180)}, {cos(180)})--({sin(240)}, {cos(240)});
        \draw [red, dashed] ({sin(240)}, {cos(240)})--({sin(300)}, {cos(300)});
        \draw [red] ({sin(0)}, {cos(0)})--({sin(300)}, {cos(300)});
        
        \draw ({sin(0)}, {cos(0)})--({1.7*sin(0)}, {1.7*cos(0)});
        \draw ({sin(60)}, {cos(60)})--({1.7*sin(60)}, {1.7*cos(60)});
        \draw [red] ({sin(120)}, {cos(120)})--({1.7*sin(120)}, {1.7*cos(120)});
        \draw ({sin(180)}, {cos(180)})--({1.7*sin(180)}, {1.7*cos(180)});
        \draw [red] ({sin(240)}, {cos(240)})--({1.7*sin(240)}, {1.7*cos(240)});
        \draw ({sin(300)}, {cos(300)})--({1.7*sin(300)}, {1.7*cos(300)});
    \end{tikzpicture}
    \; \raisebox{22pt}{=} \;
    \begin{tikzpicture}[scale = 0.5, thick]
        \draw ({sin(60)}, {cos(60)})--({sin(120)}, {cos(120)});
        \draw ({sin(120)}, {cos(120)})--({sin(180)}, {cos(180)});
        \draw [dashed] ({sin(180)}, {cos(180)})--({sin(240)}, {cos(240)});
        \draw [dashed] ({sin(240)}, {cos(240)})--({sin(300)}, {cos(300)});
        \draw ({sin(0)}, {cos(0)})--({sin(300)}, {cos(300)});
        
        \draw ({sin(0)}, {cos(0)})--({1.7*sin(0)}, {1.7*cos(0)});
        \draw ({sin(60)}, {cos(60)})--({1.7*sin(60)}, {1.7*cos(60)});
        \draw ({sin(120)}, {cos(120)})--({1.7*sin(120)}, {1.7*cos(120)});
        \draw [white] ({sin(180)}, {cos(180)})--({1.7*sin(180)}, {1.7*cos(180)});
        \draw ({sin(240)}, {cos(240)})--({1.7*sin(240)}, {1.7*cos(240)});
    \end{tikzpicture}
}\]
Note that in this example, the leaves which are omitted correspond to \(\FF = \{2,6\}\). Note \(\mathbb{E} = \{3,5\}\), \(\mathbb{O} = \{2,4\}\), \(c_\mathbb{E}(q_{001100}q_{100010}) = (1,0) \in (\zmodtwo)^{\{3,5\}}\), and \(c_\mathbb{O}(q_{001100}q_{100010}) = (0,0) \in (\zmodtwo)^{\{2,4\}}\). Putting these two colorings together gives us \((0,1,0,0) \in (\zmodtwo)^{\{2,3,4,5\}}\). We see that the 1 in the coloring indicates that \(e_{6+3}\) should be removed while the zeros in positions \(2,4\), and \(5\) indicate that the edges \(e_{6+2},e_{6+4}\), and \(e_{6+5}\) should remain in the resulting diagram.  In fact, these observations hold true as long as \(1 \notin \FF\). Therefore, we define the diagram for \(q_\bfg q_\bfh\in \cg(n,\FF,\bfa)\) (for any \(\FF\)) by omitting any leaves which are in \(\FF\) and any edge \(e_{n+k}\) when the coloring of the monomial in position \(k\) is 1. These diagrams gives us a visual interpretation of the colorings \(c_\mathbb{E}\) and \(c_\mathbb{O}\).

\begin{example}
These diagrams give us an easy way to tell if an element \(f \in \cg(n,\FF,\bfa)\) is in an invariant. For example, take \(f = q_{101111}q_{111000} - q_{101011}q_{111100} + q_{101101}q_{111010} - q_{101110}q_{111001} \in \cg(\{1,3\}, (1,1)) \subset R_6\). Here \(\mathbb{E} = \{4\}\) and \(\mathbb{O} = \{2,3,5\}\). Then \(c_\mathbb{E}\) and \(c_\mathbb{O}\) on each monomial is as follows.
\begin{align*}
    c_\mathbb{E}(q_{101000}q_{111111}) &= 0 &
    c_\mathbb{O}(q_{101000}q_{111111}) &= (1,0,0) \\
    c_\mathbb{E}(q_{101011}q_{111100}) &= 0 &
    c_\mathbb{O}(q_{101011}q_{111100}) &= (1,0,1) \\
    c_\mathbb{E}(q_{101101}q_{111010}) &= 1 &
    c_\mathbb{O}(q_{101101}q_{111010}) &= (1,0,1) \\
    c_\mathbb{E}(q_{101110}q_{111001}) &= 1 &
    c_\mathbb{O}(q_{101110}q_{111001}) &= (1,0,0)
\end{align*}
Pictorially, this is as follows: 
\[{
\psi_6(q_{101111}q_{111000} + q_{101101}q_{111010}) = \psi_6(q_{101011}q_{111100} + q_{101110}q_{111001})
}\]
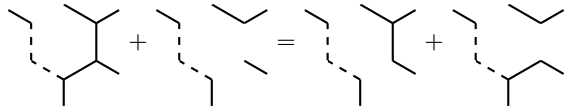
\[{
    \begin{tikzpicture}[scale = 0.5, thick]
        \draw ({sin(0)}, {cos(0)})--({sin(60)}, {cos(60)});
        \draw ({sin(60)}, {cos(60)})--({sin(120)}, {cos(120)});
        \draw ({sin(120)}, {cos(120)})--({sin(180)}, {cos(180)});
        \draw [dashed] ({sin(180)}, {cos(180)})--({sin(240)}, {cos(240)});
        \draw [dashed] ({sin(240)}, {cos(240)})--({sin(300)}, {cos(300)});
        
        \draw [white] ({sin(0)}, {cos(0)})--({1.7*sin(0)}, {1.7*cos(0)});
        \draw ({sin(60)}, {cos(60)})--({1.7*sin(60)}, {1.7*cos(60)});
        \draw ({sin(120)}, {cos(120)})--({1.7*sin(120)}, {1.7*cos(120)});
        \draw ({sin(180)}, {cos(180)})--({1.7*sin(180)}, {1.7*cos(180)});
        \draw [white] ({sin(240)}, {cos(240)})--({1.7*sin(240)}, {1.7*cos(240)});
        \draw ({sin(300)}, {cos(300)})--({1.7*sin(300)}, {1.7*cos(300)});
    \end{tikzpicture}
    \; \raisebox{22pt}{$+$} \; 
    \begin{tikzpicture}[scale = 0.5, thick]
        \draw ({sin(0)}, {cos(0)})--({sin(60)}, {cos(60)});
        \draw [dashed] ({sin(180)}, {cos(180)})--({sin(240)}, {cos(240)});
        \draw [dashed] ({sin(240)}, {cos(240)})--({sin(300)}, {cos(300)});
        
        \draw [white] ({sin(0)}, {cos(0)})--({1.7*sin(0)}, {1.7*cos(0)});
        \draw ({sin(60)}, {cos(60)})--({1.7*sin(60)}, {1.7*cos(60)});
        \draw ({sin(120)}, {cos(120)})--({1.7*sin(120)}, {1.7*cos(120)});
        \draw ({sin(180)}, {cos(180)})--({1.7*sin(180)}, {1.7*cos(180)});
        \draw [white] ({sin(240)}, {cos(240)})--({1.7*sin(240)}, {1.7*cos(240)});
        \draw ({sin(300)}, {cos(300)})--({1.7*sin(300)}, {1.7*cos(300)});
    \end{tikzpicture}
    \; \raisebox{22pt}{$=$} \; 
    \begin{tikzpicture}[scale = 0.5, thick]
        \draw ({sin(0)}, {cos(0)})--({sin(60)}, {cos(60)});
        \draw ({sin(60)}, {cos(60)})--({sin(120)}, {cos(120)});
        \draw [dashed] ({sin(180)}, {cos(180)})--({sin(240)}, {cos(240)});
        \draw [dashed] ({sin(240)}, {cos(240)})--({sin(300)}, {cos(300)});
        
        \draw [white] ({sin(0)}, {cos(0)})--({1.7*sin(0)}, {1.7*cos(0)});
        \draw ({sin(60)}, {cos(60)})--({1.7*sin(60)}, {1.7*cos(60)});
        \draw ({sin(120)}, {cos(120)})--({1.7*sin(120)}, {1.7*cos(120)});
        \draw ({sin(180)}, {cos(180)})--({1.7*sin(180)}, {1.7*cos(180)});
        \draw [white] ({sin(240)}, {cos(240)})--({1.7*sin(240)}, {1.7*cos(240)});
        \draw ({sin(300)}, {cos(300)})--({1.7*sin(300)}, {1.7*cos(300)});
    \end{tikzpicture}
    \; \raisebox{22pt}{$+$} \; 
    \begin{tikzpicture}[scale = 0.5, thick]
        \draw ({sin(0)}, {cos(0)})--({sin(60)}, {cos(60)});
        \draw ({sin(120)}, {cos(120)})--({sin(180)}, {cos(180)});
        \draw [dashed] ({sin(180)}, {cos(180)})--({sin(240)}, {cos(240)});
        \draw [dashed] ({sin(240)}, {cos(240)})--({sin(300)}, {cos(300)});
        
        \draw [white] ({sin(0)}, {cos(0)})--({1.7*sin(0)}, {1.7*cos(0)});
        \draw ({sin(60)}, {cos(60)})--({1.7*sin(60)}, {1.7*cos(60)});
        \draw ({sin(120)}, {cos(120)})--({1.7*sin(120)}, {1.7*cos(120)});
        \draw ({sin(180)}, {cos(180)})--({1.7*sin(180)}, {1.7*cos(180)});
        \draw [white] ({sin(240)}, {cos(240)})--({1.7*sin(240)}, {1.7*cos(240)});
        \draw ({sin(300)}, {cos(300)})--({1.7*sin(300)}, {1.7*cos(300)});
    \end{tikzpicture}
}\]
We can tell that \(f \in J_6 \cap \cg(\{1,3\},(1,1))\) by noting that the odd colors, \((1,0,0)\) and \((1,0,1)\), and the even colors, \(0\) and \(1\), appear once on each side of the equation, i.e. \(M_\mathbb{E}^{\{1,3\},(1,1)}(f)\) and  \(M_\mathbb{O}^{\{1,3\},(1,1)}(f)\) are both 0.

On the other hand, one can also see that \(J_6 \cap \cg(\{1,3\}, (1,1))\) contains no binomials of the form 
\[
z_1 q_{\bfg_1} q_{\bfh_1} - z_2 q_{\bfg_2} q_{\bfh_2}
\]
for any suitable group elements and complex numbers \(z_i \in \Cc \setminus\{0\}\). The reason being that if this were to vanish under \(\psi_6\) that would mean that \(z_1 = z_2\) and the colorings of each \(q_{\bfg_i}q_{\bfh_i}\) would need to be identical, but this would imply that \(\bfg_1 = \bfg_2\) and \(\bfh_1 = \bfh_2\).
\end{example}

\begin{example}
\label{ex:6sunlet}
Let \(n = 6\) and \(\FF = \emptyset\). In this case, \(\cg(6,\emptyset)\) is spanned by the following 16 monomials.
\[{
q_{000000}q_{111111}, q_{000011}q_{111100}, q_{000101}q_{111010}, q_{000110}q_{111001},
}\]
\[{
q_{001001}q_{110110}, q_{001010}q_{110101}, q_{001100}q_{110011}, q_{001111}q_{110000}, 
}\] 
\[{
q_{010001}q_{101110}, q_{010010}q_{101101}, q_{010100}q_{101011}, q_{010111}q_{101000}, 
}\]
\[{
q_{011000}q_{100111}, q_{011011}q_{100100}, q_{011101}q_{100010}, q_{011110}q_{100001}
}\]

We have the following matrices where the columns are indexed by the monomials above and the rows are indexed by elements of \((\zmodtwo)^2\) lexicographically.
\[
M_\mathbb{E}^{6,\emptyset} = {
    \left(
        {\begin{array}{cccccccccccccccc}
        1 & 1 & 0 & 0 & 0 & 0 & 1 & 1 & 0 & 0 & 0 & 0 & 0 & 0 & 0 & 0 \\
        0 & 0 & 1 & 1 & 1 & 1 & 0 & 0 & 0 & 0 & 0 & 0 & 0 & 0 & 0 & 0 \\
        0 & 0 & 0 & 0 & 0 & 0 & 0 & 0 & 0 & 0 & 1 & 1 & 1 & 1 & 0 & 0 \\
        0 & 0 & 0 & 0 & 0 & 0 & 0 & 0 & 1 & 1 & 0 & 0 & 0 & 0 & 1 & 1\\
        \end{array}}
    \right)
}
\]
\[
M_\mathbb{O}^{6,\emptyset} = {
    \left(
        {\begin{array}{cccccccccccccccc}
        1 & 0 & 0 & 1 & 0 & 0 & 0 & 0 & 0 & 0 & 0 & 0 & 1 & 0 & 0 & 1 \\
        0 & 1 & 1 & 0 & 0 & 0 & 0 & 0 & 0 & 0 & 0 & 0 & 0 & 1 & 1 & 0 \\
        0 & 0 & 0 & 0 & 0 & 1 & 1 & 0 & 0 & 1 & 1 & 0 & 0 & 0 & 0 & 0 \\
        0 & 0 & 0 & 0 & 1 & 0 & 0 & 1 & 1 & 0 & 0 & 1 & 0 & 0 & 0 & 0\\
        \end{array}}
    \right)
}
\]

Then \(J_6 \cap \cg(6,\emptyset)\) is a 9 dimensional \(\Cc\)-vector space spanned by the following polynomials.
\[
q_{000000}q_{111111} - q_{000110}q_{111001} + q_{000101}q_{111010} - q_{000011}q_{111100}
\]
\[
q_{000000}q_{111111} - q_{000110}q_{111001} + 
q_{001010}q_{110101} - q_{001100}q_{110011}
\]
\[
q_{000000}q_{111111} - q_{000110}q_{111001} + 
q_{001001}q_{110110} - q_{001111}q_{110000}
\]
\[
q_{000000}q_{111111} - q_{011000}q_{100111} + 
q_{011011}q_{100100} - q_{000011}q_{111100}
\]
\[
q_{000000}q_{111111} - q_{011000}q_{100111} + 
q_{010100}q_{101011} - q_{001100}q_{110011}
\]
\[
q_{000000}q_{111111} - q_{011000}q_{100111} + 
q_{010111}q_{101000} - q_{001111}q_{110000}
\]
\[
q_{000000}q_{111111} - q_{011110}q_{100001} +
q_{011101}q_{100010} - q_{000011}q_{111100}
\]
\[
q_{000000}q_{111111} - q_{011110}q_{100001} +
q_{010010}q_{101101} - q_{001100}q_{110011}
\]
\[
q_{000000}q_{111111} - q_{011110}q_{100001} +
q_{010001}q_{101110} - q_{001111}q_{110000}
\]

Let us consider the colorings of the monomials in the last polynomial. Note \(\mathbb{E} = \{2,4\}\) and \(\mathbb{O} = \{3,5\}\)
\begin{align*}
c_\mathbb{E}(q_{000000}q_{111111}) &= (0,0) &
c_\mathbb{O}(q_{000000}q_{111111}) &= (0,0) \\
c_\mathbb{E}(q_{011110}q_{100001}) &= (1,1) &
c_\mathbb{O}(q_{011110}q_{100001}) &= (0,0) \\
c_\mathbb{E}(q_{010001}q_{101110}) &= (1,1) &
c_\mathbb{O}(q_{010001}q_{101110}) &= (1,1) \\
c_\mathbb{E}(q_{001111}q_{110000}) &= (0,0) &
c_\mathbb{O}(q_{001111}q_{110000}) &= (1,1) 
\end{align*}
This relation can be viewed pictorially as 
\[{
\psi_6(q_{000000}q_{111111} + q_{010001}q_{101110}) = \psi_6(q_{011110}q_{100001} + q_{001111}q_{110000})
}\]
\[{
    \begin{tikzpicture}[scale = 0.5, thick]
        \draw  ({sin(0)}, {cos(0)})--({sin(60)}, {cos(60)});
        \draw  ({sin(60)}, {cos(60)})--({sin(120)}, {cos(120)});
        \draw  ({sin(120)}, {cos(120)})--({sin(180)}, {cos(180)});
        \draw [dashed] ({sin(180)}, {cos(180)})--({sin(240)}, {cos(240)});
        \draw [dashed] ({sin(240)}, {cos(240)})--({sin(300)}, {cos(300)});
        \draw ({sin(0)}, {cos(0)})--({sin(300)}, {cos(300)});
        
        \draw ({sin(0)}, {cos(0)})--({1.7*sin(0)}, {1.7*cos(0)});
        \draw ({sin(60)}, {cos(60)})--({1.7*sin(60)}, {1.7*cos(60)});
        \draw ({sin(120)}, {cos(120)})--({1.7*sin(120)}, {1.7*cos(120)});
        \draw ({sin(180)}, {cos(180)})--({1.7*sin(180)}, {1.7*cos(180)});
        \draw ({sin(240)}, {cos(240)})--({1.7*sin(240)}, {1.7*cos(240)});
        \draw ({sin(300)}, {cos(300)})--({1.7*sin(300)}, {1.7*cos(300)});
        
    \end{tikzpicture} 
    \; \raisebox{22pt}{+} \;
    \begin{tikzpicture}[scale = 0.5, thick]
        \draw [dashed] ({sin(180)}, {cos(180)})--({sin(240)}, {cos(240)});
        \draw [dashed] ({sin(240)}, {cos(240)})--({sin(300)}, {cos(300)});
        
        \draw ({sin(0)}, {cos(0)})--({1.7*sin(0)}, {1.7*cos(0)});
        \draw ({sin(60)}, {cos(60)})--({1.7*sin(60)}, {1.7*cos(60)});
        \draw ({sin(120)}, {cos(120)})--({1.7*sin(120)}, {1.7*cos(120)});
        \draw ({sin(180)}, {cos(180)})--({1.7*sin(180)}, {1.7*cos(180)});
        \draw ({sin(240)}, {cos(240)})--({1.7*sin(240)}, {1.7*cos(240)});
        \draw ({sin(300)}, {cos(300)})--({1.7*sin(300)}, {1.7*cos(300)});
        
        
    \end{tikzpicture}
    \; \raisebox{22pt}{=} \;
    \begin{tikzpicture}[scale = 0.5, thick]
        \draw ({sin(0)}, {cos(0)})--({sin(60)}, {cos(60)});
        \draw ({sin(120)}, {cos(120)})--({sin(180)}, {cos(180)});
        \draw [dashed] ({sin(180)}, {cos(180)})--({sin(240)}, {cos(240)});
        \draw [dashed] ({sin(240)}, {cos(240)})--({sin(300)}, {cos(300)});
        
        \draw ({sin(0)}, {cos(0)})--({1.7*sin(0)}, {1.7*cos(0)});
        \draw ({sin(60)}, {cos(60)})--({1.7*sin(60)}, {1.7*cos(60)});
        \draw ({sin(120)}, {cos(120)})--({1.7*sin(120)}, {1.7*cos(120)});
        \draw ({sin(180)}, {cos(180)})--({1.7*sin(180)}, {1.7*cos(180)});
        \draw ({sin(240)}, {cos(240)})--({1.7*sin(240)}, {1.7*cos(240)});
        \draw ({sin(300)}, {cos(300)})--({1.7*sin(300)}, {1.7*cos(300)});
        
    \end{tikzpicture}
    \; \raisebox{22pt}{+} \;
    \begin{tikzpicture}[scale = 0.5, thick]
        \draw ({sin(60)}, {cos(60)})--({sin(120)}, {cos(120)});
        \draw [dashed] ({sin(180)}, {cos(180)})--({sin(240)}, {cos(240)});
        \draw [dashed] ({sin(240)}, {cos(240)})--({sin(300)}, {cos(300)});
        \draw  ({sin(0)}, {cos(0)})--({sin(300)}, {cos(300)});
        
        \draw ({sin(0)}, {cos(0)})--({1.7*sin(0)}, {1.7*cos(0)});
        \draw ({sin(60)}, {cos(60)})--({1.7*sin(60)}, {1.7*cos(60)});
        \draw ({sin(120)}, {cos(120)})--({1.7*sin(120)}, {1.7*cos(120)});
        \draw ({sin(180)}, {cos(180)})--({1.7*sin(180)}, {1.7*cos(180)});
        \draw ({sin(240)}, {cos(240)})--({1.7*sin(240)}, {1.7*cos(240)});
        \draw ({sin(300)}, {cos(300)})--({1.7*sin(300)}, {1.7*cos(300)});
        
    \end{tikzpicture}
}\]

We also note that the dimension of \(\cs_6\) is \(12\), its codimension is 20, and \(J_6\) is minimally generated by 79 polynomials; thus, contrary to say the 4-leaf case, \(\cs_6\) is not a complete intersection even set-theoretically.
\end{example}

\subsection{Quadratic phylogenetic invariants of trees}
\label{sec:TreeGloves}


Let \(\ct\) be a binary tree with leaf set \([n]\), and let \(I_\ct\) be the defining ideal for the corresponding variety. As was discussed in Section \ref{sec:cfnModel}, the phylogenetic invariants for this model are given purely in terms of \(2\times 2\) minors of certain matrices. In this section, we give a separate description for the generating set which is in line with the approach from Section \ref{sec:NetworkGloves}. Using the same reasoning as in Section \ref{sec:Zgrading}, we can see that \(I_\ct\) is also graded by \(\Zz^{n+1}\); hence, the quadratic generators can be described by the \(\Cc\)-vector spaces \(I_\ct \cap \cg(n,\FF,\bfa)\) where we can again restrict to when \([n] \sm \FF\) is even has cardinality at least 4. Recall that for any edge \(e \in \Sigma(\ct)\), the edge induces a split of the tree \(A_e \vert B_e\). In this section, given a glove \(\cg(n,\FF,\bfa)\), we define 
\[
    \mathbb{E}_\ct(\FF) = \{e \in E(\ct) \vert |A_e \setminus \FF| \text{ is even}\}.
\]
When it is clear from context, we will simply write \(\mathbb{E}_\ct\). Similarly, we let the linear map \(M_{\mathbb{E}_\ct}^{\FF,\bfa} : \cg(n,\FF,\bfa) \to \Cc^{(\zmodtwo)^{\mathbb{E}_\ct}}\) be defined by the following matrix as in the previous subsection.
\[
(M_{\mathbb{E}_\ct}^{n, \FF,\bfa})_{\bfc, q_\bfg q_\bfh} = 
    \begin{cases}
    1 & \text{if for all }e \in \mathbb{E}_\ct \text{, } c_e = \sum_{i \in A_e} g_i \\
    0 & \text{otherwise}
    \end{cases}
\]
where \(\bfg <_\text{lex} \bfh\). Then we have the following theorem which is analogous to Theorem \ref{thm:MainThm}, but for trees.

\begin{theorem}
\label{thm:TreeGloves}
Given a glove \(\cg(n,\FF,\bfa)\) and a phylogenetic tree \(\ct\), the \(\Zz^{n+1}\)-graded piece \(I_\ct \cap \cg(n,\FF,\bfa)\) is the kernel of \(M_{\mathbb{E}_\ct}^{n, \FF,\bfa}\).
\end{theorem}
\begin{proof}
Let \(S_\ct = \Cc[a_g^e ~|~ g \in \zmodtwo \text{ and }e \in E(\ct)]\). Recall that \(I_\ct\) is the kernel of \(\psi_T : R_n \to S_\ct\) defined by 
\[
    q_\bfg \mapsto\prod_{A_e|B_e \in \Sigma(T)} a_{\sum_{i \in A_e} g_i}^e
\]
Now, fix a glove \(\cg(n,\FF,\bfa)\), and note that if \(q_\bfg q_\bfh \in \cg(n,\FF,\bfa)\), then \(\sum_{i \in A_e} g_i = \sum_{i \in A_e} h_i\) if and only if \(e \in \mathbb{E}_\ct\). Consider any polynomial \(f = \sum_{\bfg \in L(n,\FF,\bfa)} c_\bfg q_\bfg q_\bfh \in \cg(n,\FF,\bfa)\). If we apply \(\psi_T\), we get the following.
\begin{align*}
    \psi_\ct(f) &= \sum_{\bfg \in L(n,\FF,\bfa)} c_\bfg \left(\prod_{e \in E(\ct)} a_{\sum_{i\in A_e} g_i}^e \right) \left(\prod_{e \in E(\ct)} a_{\sum_{i\in A_e} h_i}^e \right) \\
            &= \sum_{\bfg \in L(n,\FF,\bfa)} c_\bfg\left( \prod_{e \notin \mathbb{E}} a_0^e a_1^e  \right) \left(\prod_{e \in \mathbb{E}_\ct} (a_{\sum_{i\in A_e} g_i}^e)^2 \right) \\
            &= \left( \prod_{e \notin \mathbb{E}_\ct} a_0^e a_1^e  \right)\sum_{\bfg \in L(n,\FF,\bfa)} c_\bfg \left(\prod_{e \in \mathbb{E}} (a_{\sum_{i\in A_e} g_i}^e)^2 \right)
\end{align*}
The monomials, \(\prod_{e \in \mathbb{E}_\ct} (a_{\sum_{i\in A_e} g_i}^e)^2 \), can be identified as standard basis vectors in \(\Cc^{(\zmodtwo)^{\mathbb{E}_\ct}}\). After making this identification, it becomes evident that \(\psi_\ct(f) = 0\) if and only if \(M_{\mathbb{E}_\ct}^{\FF,\bfa}(f) = 0\).
\end{proof}

Consider \(\cs_n\) and its two underlying trees \(\ct_0\) and \(\ct_1\), and fix any glove \(\cg(n,\FF,\bfa)\) where either \(1\) is not in \(\FF\) or \(1\) is in \(\FF\) but \(a_1 = 1\). Recall that \(\ct_0\) is obtained by deleting the reticulation edge that lies between the leaves \(e_1\) and \(e_2\), and \(\ct_1\) is obtained by deleting the reticulation edge that lies between the leaves \(e_1\) and \(e_n\). The defining ideals for \(\ct_0\) and \(\ct_1\) are generated by quadratic binomials. Here we will show that the polynomials \(f_\bfc\) from Proposition \ref{prop:idealgenerators} are either sums or differences of binomials coming from \(I_{\ct_0}\) and \(I_{\ct_1}\). In the following proposition, we only consider the case when \(n\) is even, at least 4, and \(\FF = \emptyset\) since any other glove of the form stated can be obtained from this case.

\begin{proposition}
\label{prop:sumsofminors}
Let \(n \in \Zz_{\geq 4}\) be even, and consider any polynomial 
\[
    f_\bfc = q_{\bfg({\bf 0}, {\bf 0})}q_{\bfh({\bf 0}, {\bf 0})} - q_{\bfg(\bfc\vert_\mathbb{E}, {\bf 0})}q_{\bfh(\bfc\vert_\mathbb{E}, {\bf 0})} + q_{\bfg(\bfc\vert_\mathbb{E}, \bfc\vert_\mathbb{O})}q_{\bfh(\bfc\vert_\mathbb{E}, \bfc\vert_\mathbb{O})} - q_{\bfg({\bf 0}, \bfc\vert_\mathbb{O})} q_{\bfh({\bf 0}, \bfc\vert_\mathbb{O})}
\]
in \(J_n \cap \cg(n,\emptyset)\) from Proposition \ref{prop:idealgenerators}. Then 
\[
     q_{\bfg(\bfc\vert_\mathbb{E}, {\bf 0})}q_{\bfh(\bfc\vert_\mathbb{E}, {\bf 0})} - q_{\bfg(\bfc\vert_\mathbb{E}, \bfc\vert_\mathbb{O})}q_{\bfh(\bfc\vert_\mathbb{E}, \bfc\vert_\mathbb{O})}
\]
\[   
     q_{\bfg({\bf 0}, {\bf 0})}q_{\bfh({\bf 0}, {\bf 0})} - q_{\bfg({\bf 0}, \bfc\vert_\mathbb{O})} q_{\bfh({\bf 0}, \bfc\vert_\mathbb{O})}
\]
are in \(I_{\ct_0} \cap \cg(n, \emptyset)\), and 
\[
    q_{\bfg({\bf 0}, {\bf 0})}q_{\bfh({\bf 0}, {\bf 0})} - q_{\bfg(\bfc\vert_\mathbb{E}, {\bf 0})}q_{\bfh(\bfc\vert_\mathbb{E}, {\bf 0})}
\]
\[
    q_{\bfg(\bfc\vert_\mathbb{E}, \bfc\vert_\mathbb{O})}q_{\bfh(\bfc\vert_\mathbb{E}, \bfc\vert_\mathbb{O})} - q_{\bfg({\bf 0}, \bfc\vert_\mathbb{O})} q_{\bfh({\bf 0}, \bfc\vert_\mathbb{O})} 
\]
are in \(I_{\ct_1} \cap \cg(n,\emptyset)\). \(\mathbb{E}\) is the set of even numbers between 2 and \(n-1\) and \(\mathbb{O}\) is the odd numbers in the same range.
\end{proposition}

\begin{proof}
Note that \(\mathbb{E} = \mathbb{E}_{\ct_0}\) and \(\mathbb{O} = \mathbb{E}_{\ct_1}\). Then the claim follows by Theorem \ref{thm:TreeGloves}.
\end{proof}

\subsection{Proof of Theorem \ref{thm:MainThm}}
\label{sec:GlovesProof}

Let \(f = \sum_{q_\bfg q_\bfh \in \mathcal{G}(n, \FF,\bfa)} c_{\bfg,\bfh} q_\bfg q_\bfh\) with \(c_{\bfg,\bfh} \in \Cc\). We want to identify necessary and sufficient conditions on the coefficients \(c_{\bfg,\bfh}\) for \(f \in J_n\). We analyze the gloves in three cases.
\begin{enumerate}
    \item \(1 \notin \FF\)
    \item \(1 \in \FF\) and \(a_1 = 1\)
    \item \(1 \in \FF\) and \(a_1 = 0\)
\end{enumerate}

{\bf Case 1: \(1 \notin \FF\).} First, note that for each monomial, \(q_\bfg q_\bfh\), in \(\mathcal{G}(n, \FF,\bfa)\), either \(g_1\) or \(h_1\) is 0. We will always assume that \(g_1 = 0\), and we also remark that \(\bfh\) is completely determined by \(\bfg\); therefore, we will write \(c_\bfg\) instead of \(c_{\bfg,\bfh}\). We have the set \(L(n,\FF,\bfa)\) as defined in Section \ref{sec:gloves} which in this case simplifies to 
\[
L(n,\FF,\bfa) = \{\bfg \in (\zmodtwo)^n ~|~ \text{there exists } \bfh \text{ so that } q_\bfg q_\bfh \in \mathcal{G}(n, \FF,\bfa) \text{ and } g_1  = 0\}.
\]
Now, we compute \(\psi_n(f)\).
\begin{align*}
    \psi_n(f) &= \sum_{\bfg \in L(n,\FF,\bfa)} c_\bfg \psi_n (q_\bfg q_\bfh) \\
                 &= \sum_{\bfg \in L(n,\FF,\bfa)} c_\bfg \left(\prod_{j=1}^n a_{g_j}^j \right)\left(\prod_{j=1}^{n-1} a_{\sum_{\ell =1}^{j} g_\ell}^{n+j}  + \prod_{j=2}^{n} a_{\sum_{\ell =2}^{j} g_\ell}^{n+j}\right) \left( \prod_{j=1}^n a_{h_j}^j \right) \left(\prod_{j=1}^{n-1} a_{\sum_{\ell =1}^{j} h_\ell}^{n+j}  + \prod_{j=2}^{n} a_{\sum_{\ell =2}^{j} h_\ell}^{n+j} \right)\\
                 &= \sum_{\bfg \in L(n,\FF,\bfa)} c_\bfg 
                 \left(\prod_{j \in \FF} (a_{g_j}^j)^2\right)
                 \left(\prod_{j \notin \FF} a_0^j a_1^j \right)
                 \left(\prod_{j=1}^{n-1} a_{\sum_{\ell =1}^{j} g_\ell}^{n+j}  + \prod_{j=2}^{n} a_{\sum_{\ell =2}^{j} g_\ell}^{n+j}\right)
                 \left(\prod_{j=1}^{n-1} a_{\sum_{\ell =1}^{j} h_\ell}^{n+j}  + \prod_{j=2}^{n} a_{\sum_{\ell =2}^{j} h_\ell}^{n+j}\right)
\end{align*}
The monomial \(\left(\prod_{j \in \FF} (a_{g_j}^j)^2\right)\left(\prod_{j \notin \FF} a_0^j a_1^j \right)\) depends only on \(\FF\) and \(\bfa\), so this can be factored out and denoted by \(m_{\FF,\bfa}\).
\begin{align*}
   \psi_n(f) &= m_{\FF,\bfa} \sum_{\bfg \in L(n,\FF,\bfa)} c_\bfg 
                \left(\prod_{j=1}^{n-1} a_{\sum_{\ell =1}^{j} g_\ell}^{n+j}  + \prod_{j=2}^{n} a_{\sum_{\ell =2}^{j} g_\ell}^{n+j}\right)
                \left(\prod_{j=1}^{n-1} a_{\sum_{\ell =1}^{j} h_\ell}^{n+j}  + \prod_{j=2}^{n} a_{\sum_{\ell =2}^{j} h_\ell}^{n+j}\right) \\
                &= m_{\FF,\bfa} \sum_{\bfg \in L(n,\FF,\bfa)} c_\bfg 
                \left(a_{g_1}^{n+1}\prod_{j=2}^{n-1} a_{\sum_{\ell =1}^{j} g_\ell}^{n+j}  + a_{g_1}^{2n}\prod_{j=2}^{n-1} a_{\sum_{\ell =2}^{j} g_\ell}^{n+j}\right)
                \left(a_{h_1}^{n+1}\prod_{j=2}^{n-1} a_{\sum_{\ell =1}^{j} h_\ell}^{n+j}  + a_{h_1}^{2n}\prod_{j=2}^{n-1} a_{\sum_{\ell =2}^{j} h_\ell}^{n+j}\right) \\
                &= m_{\FF,\bfa} \sum_{\bfg \in L(n,\FF,\bfa)} c_\bfg 
                \left(a_{0}^{n+1}\prod_{j=2}^{n-1} a_{\sum_{\ell =1}^{j} g_\ell}^{n+j}  + a_{0}^{2n}\prod_{j=2}^{n-1} a_{\sum_{\ell =2}^{j} g_\ell}^{n+j}\right)
                \left(a_{1}^{n+1}\prod_{j=2}^{n-1} a_{\sum_{\ell =1}^{j} h_\ell}^{n+j}  + a_{1}^{2n}\prod_{j=2}^{n-1} a_{\sum_{\ell =2}^{j} h_\ell}^{n+j}\right)
\end{align*}
We proceed by multiplying these two binomials and make the following observations about the various sums in the subscripts.
\begin{itemize}
    \item \(\sum_{\ell = 1}^j g_\ell  = \sum_{\ell = 2}^j g_\ell\) since \(g_1 = 0\)
    \item \(\sum_{\ell = 1}^j g_\ell = \sum_{\ell = 1}^j h_\ell\) if and only if \([j] \sm \FF\) has even cardinality.
    \item \(\sum_{\ell = 2}^j g_\ell = \sum_{\ell = 2}^j h_\ell\) if and only if \([j] \sm \FF\) has odd cardinality.
\end{itemize}
This yields the following.
\begin{align*}
    \psi_n(f) &= m_{\FF,\bfa} \sum_{\bfg \in L(n,\FF,\bfa)} c_\bfg
                 \left( a_0^{n+1}a_1^{n+1} \prod_{j \in \mathbb{E}} (a_{\sum_{\ell = 1}^{j} g_\ell}^{n+j})^2 
                 \prod_{j \in \mathbb{O}} a_0^{n+j}a_1^{n+j} \right.\\   
                 &\;\;\;\;\;\;\;\;\;\;\;\;\;\;\;\;\;\;\;\;+ a_0^{n+1}a_{1}^{2n} \prod_{j \in \mathbb{E}} a_0^{n+j} a_1^{n+j}
                 \prod_{j \in \mathbb{O}} (a_{\sum_{\ell=1}^j g_\ell}^{n+j})^2 \\
                 &\;\;\;\;\;\;\;\;\;\;\;\;\;\;\;\;\;\;\;\;+ a_1^{n+1}a_{0}^{2n} \prod_{j \in \mathbb{E}} (a_{\sum_{\ell =1}^j g_\ell}^{n+j})^2 
                 \prod_{j \in \mathbb{O}} a_0^{n+j} a_1^{n+j} \\    
                 &\;\;\;\;\;\;\;\;\;\;\;\;\;\;\;\;\;\;\;\;\left.+ a_{0}^{2n}a_{1}^{2n} \prod_{j \in \mathbb{E}} a_0^{n+j}a_1^{n+j}
                 \prod_{j \in \mathbb{O}} (a_{\sum_{\ell=1}^j g_\ell}^{n+j})^2
                 \right) 
\end{align*}
Note that the following products depend only on \(\FF\). We make these substitutions and proceed.
\begin{align*}
m_{\FF,\mathbb{E}} & := \prod_{j \in \mathbb{E}} a_0^{n+j}a_1^{n+j} \\
m_{\FF,\mathbb{O}} & := \prod_{j \in \mathbb{O}} a_0^{n+j}a_1^{n+j}
\end{align*}

\begin{align*}
    \psi_n(f) &= m_{\FF,\bfa} \sum_{\bfg \in L(n,\FF,\bfa)} c_\bfg
                 \left( a_0^{n+1}a_1^{n+1}m_{\FF,\mathbb{O}} \prod_{j \in \mathbb{E}} (a_{\sum_{\ell = 1}^{j} g_\ell}^{n+j})^2 \right.\\   
                 &\;\;\;\;\;\;\;\;\;\;\;\;\;\;\;\;\;\;\;\;+ a_0^{n+1}a_{1}^{2n} m_{\FF,\mathbb{E}}
                 \prod_{j \in \mathbb{O}} (a_{\sum_{\ell=1}^j g_\ell}^{n+j})^2 +
                  a_1^{n+1}a_{0}^{2n} m_{\FF,\mathbb{O}} \prod_{j \in \mathbb{E}} (a_{\sum_{\ell =1}^j g_\ell}^{n+j})^2 \\    
                 &\;\;\;\;\;\;\;\;\;\;\;\;\;\;\;\;\;\;\;\;\left.+ a_{0}^{2n}a_{1}^{2n} m_{\FF,\mathbb{E}}
                 \prod_{j \in \mathbb{O}} (a_{\sum_{\ell=1}^j g_\ell}^{n+j})^2
                 \right)    \\
                 &= m_{\FF,\bfa} \sum_{g \in L(n,\FF,\bfa)} c_\bfg \left( (a_0^{n+1}a_1^{n+1} + a_1^{n+1}a_{0}^{2n})m_{\FF,\mathbb{O}}
                 \prod_{j \in \mathbb{E}} (a_{\sum_{\ell = 1}^{j} g_\ell}^{n+j})^2\right. \\
                 &\;\;\;\;\;\;\;\;\;\;\;\;\;\;\;\;\;\;\;\;\left.+ (a_0^{n+1}a_{1}^{2n} + a_{0}^{2n}a_{1}^{2n})m_{\FF,\mathbb{E}} 
                 \prod_{j \in \mathbb{O}} (a_{\sum_{\ell=1}^j g_\ell}^{n+j})^2\right) \\
                 &= m_{\FF,\bfa}m_{\FF,\mathbb{O}}(a_0^{n+1}a_1^{n+1} + a_1^{n+1}a_{0}^{2n}) \sum_{\bfg \in L(n,\FF,\bfa)} c_\bfg \prod_{j \in \mathbb{E}} (a_{\sum_{\ell = 1}^{j} g_\ell}^{n+j})^2 \\
                 &\;\;\;\;\;\;\;\;\;\;\;\;\;\;\;\;\;\;\;\;+ 
                 m_{\FF,\bfa} m_{\FF,e}(a_0^{n+1}a_{1}^{2n} + a_{0}^{2n}a_{1}^{2n}) \sum_{\bfg \in L(n,\FF,\bfa)} c_{\bfg} 
                 \prod_{j \in \mathbb{O}} (a_{\sum_{\ell=1}^j g_\ell}^{n+j})^2
\end{align*}

In the last line, we note that the superscripts appearing in the sums are completely disjoint. Since the \(c_\bfg \in \Cc\) for every \(\bfg \in L(n,\FF,\bfa)\), the only way for \(\psi_n(f) = 0\) is if both sums vanish. Recall the maps of \(M_\mathbb{E}^{n, \FF,\bfa}\) and \(M_\mathbb{O}^{n, \FF,\bfa}\) from Section \ref{sec:gloves}. By the definition of \(M_\mathbb{E}^{n, \FF,\bfa}\) the first sum vanishes if and only if \(f \in \ker M_\mathbb{E}^{n,\FF,\bfa}\), and similarly the second sum vanishes if and only if \(f \in \ker M_\mathbb{O}^{n, \FF,\bfa}\). It then follows that \(f \in J_n \cap \cg(n,\FF,\bfa)\) if and only if \(f\) lies in the intersection of these two kernels.

{\bf Case 2: \(1 \in \FF\) and \(a_1 = 1\).} First, note that for each monomial, \(q_\bfg q_\bfh\), in \(\mathcal{G}(n, \FF,\bfa)\), both \(g_1\) and \(h_1\) are 1. We will always assume that \(\bfg <_\text{lex} \bfh\), i.e. \(\bfg \in L(n,\FF,\bfa)\). Again, we remark that \(\bfh\) is completely determined by \(\bfg\); therefore, we will write \(c_\bfg\) instead of \(c_{\bfg,\bfh}\). Now, we compute \(\psi_n(f)\). 
\begin{align*}
    \psi_n(f) &= \sum_{\bfg \in L(n,\FF,\bfa)} c_\bfg \psi_n (q_\bfg q_\bfh) \\
                 &= \sum_{\bfg \in L(n,\FF,\bfa)} c_\bfg \left(\prod_{j=1}^n a_{g_j}^j \right)\left(\prod_{j=1}^{n-1} a_{\sum_{\ell =1}^{j} g_\ell}^{n+j}  + \prod_{j=2}^{n} a_{\sum_{\ell =2}^{j} g_\ell}^{n+j}\right) \left( \prod_{j=1}^n a_{h_j}^j \right) \left(\prod_{j=1}^{n-1} a_{\sum_{\ell =1}^{j} h_\ell}^{n+j}  + \prod_{j=2}^{n} a_{\sum_{\ell =2}^{j} h_\ell}^{n+j} \right)\\
                 &= \sum_{\bfg \in L(n,\FF,\bfa)} c_\bfg 
                 \left(\prod_{j \in \FF} (a_{g_j}^j)^2\right)
                 \left(\prod_{j \notin \FF} a_0^j a_1^j \right)
                 \left(\prod_{j=1}^{n-1} a_{\sum_{\ell =1}^{j} g_\ell}^{n+j}  + \prod_{j=2}^{n} a_{\sum_{\ell =2}^{j} g_\ell}^{n+j}\right)
                 \left(\prod_{j=1}^{n-1} a_{\sum_{\ell =1}^{j} h_\ell}^{n+j}  + \prod_{j=2}^{n} a_{\sum_{\ell =2}^{j} h_\ell}^{n+j}\right)
\end{align*}
The monomial \(\left(\prod_{j \in \FF} (a_{g_j}^j)^2\right) \left(\prod_{j \notin \FF} a_0^j a_1^j \right)\) depends only on \(\FF\) and \(\bfa\), so it can be factored out of the sum, and it will be denoted as \(m_{\FF,\bfa}\).
\begin{align*}
   \psi_n(f) &= m_{\FF,\bfa} \sum_{\bfg \in L(n,\FF,\bfa)} c_\bfg 
                \left(\prod_{j=1}^{n-1} a_{\sum_{\ell =1}^{j} g_\ell}^{n+j}  + \prod_{j=2}^{n} a_{\sum_{\ell =2}^{j} g_\ell}^{n+j}\right)
                \left(\prod_{j=1}^{n-1} a_{\sum_{\ell =1}^{j} h_\ell}^{n+j}  + \prod_{j=2}^{n} a_{\sum_{\ell =2}^{j} h_\ell}^{n+j}\right) \\
                &= m_{\FF,\bfa} \sum_{\bfg \in L(n,\FF,\bfa)} c_\bfg 
                \left(a_{g_1}^{n+1}\prod_{j=2}^{n-1} a_{\sum_{\ell =1}^{j} g_\ell}^{n+j}  + a_{g_1}^{2n}\prod_{j=2}^{n-1} a_{\sum_{\ell =2}^{j} g_\ell}^{n+j}\right)
                \left(a_{h_1}^{n+1}\prod_{j=2}^{n-1} a_{\sum_{\ell =1}^{j} h_\ell}^{n+j}  + a_{h_1}^{2n}\prod_{j=2}^{n-1} a_{\sum_{\ell =2}^{j} h_\ell}^{n+j}\right) \\
                &= m_{\FF,\bfa} \sum_{\bfg \in L(n,\FF,\bfa)} c_\bfg 
                \left(a_{1}^{n+1}\prod_{j=2}^{n-1} a_{\sum_{\ell =1}^{j} g_\ell}^{n+j}  + a_{1}^{2n}\prod_{j=2}^{n-1} a_{\sum_{\ell =2}^{j} g_\ell}^{n+j}\right)
                \left(a_{1}^{n+1}\prod_{j=2}^{n-1} a_{\sum_{\ell =1}^{j} h_\ell}^{n+j}  + a_{1}^{2n}\prod_{j=2}^{n-1} a_{\sum_{\ell =2}^{j} h_\ell}^{n+j}\right)
\end{align*}
Now, we will proceed by multiplying all these terms out and regrouping using the following observations about the various sums in the subscripts.
\begin{itemize}
    \item \(\sum_{\ell = 1}^j g_\ell  = 1 + \sum_{\ell = 2}^j g_\ell\) since \(g_1 = 1\)
    \item \(\sum_{\ell = 1}^j g_\ell = \sum_{\ell = 1}^j h_\ell\) if and only if \([j] \sm \FF\) has even cardinality.
    \item \(\sum_{\ell = 2}^j g_\ell = \sum_{\ell = 2}^j h_\ell\) if and only if \([j] \sm \FF\) has even cardinality.
    \item \(\sum_{\ell = 1}^j g_\ell = \sum_{\ell = 2}^j h_\ell\) if and only if \([j] \sm \FF\) has odd cardinality.
    \item \(\sum_{\ell = 2}^j g_\ell = \sum_{\ell = 1}^j h_\ell\) if and only if \([j] \sm \FF\) has odd cardinality.
\end{itemize}
Then we get the following.
\begin{align*}
    \psi_n(f) &= m_{\FF,\bfa} \sum_{\bfg \in L(n,\FF,\bfa)} c_\bfg
                 \left( a_1^{n+1}a_1^{n+1} \prod_{j \in \mathbb{E}} (a_{\sum_{\ell = 1}^{j} g_\ell}^{n+j})^2 
                 \prod_{j \in \mathbb{O}} a_0^{n+j}a_1^{n+j} \right.\\   
                 &\;\;\;\;\;\;\;\;\;\;\;\;\;\;\;\;\;\;\;\;+ a_1^{n+1}a_{1}^{2n} \prod_{j \in \mathbb{E}} a_0^{n+j} a_1^{n+j}
                 \prod_{j \in \mathbb{O}} (a_{\sum_{\ell=1}^j g_\ell}^{n+j})^2 \\
                 &\;\;\;\;\;\;\;\;\;\;\;\;\;\;\;\;\;\;\;\;+ a_1^{n+1}a_{1}^{2n} \prod_{j \in \mathbb{E}} a_0^{n+j} a_1^{n+j}  
                 \prod_{j \in \mathbb{O}} (a_{\sum_{\ell=2}^j g_\ell}^{n+j})^2 \\    
                 &\;\;\;\;\;\;\;\;\;\;\;\;\;\;\;\;\;\;\;\;\left.+ a_{1}^{2n}a_{1}^{2n} \prod_{j \in \mathbb{E}} (a_{\sum_{\ell=2}^j g_\ell}^{n+j})^2
                 \prod_{j \in \mathbb{O}} a_0^{n+j}a_1^{n+j}
                 \right) 
\end{align*}
The following products depend only on \(\FF\), so we give them names. 
\begin{align*}
m_{\FF,\mathbb{E}} &:= \prod_{j \in \mathbb{E}} a_0^{n+j}a_1^{n+j} \\ 
m_{\FF,\mathbb{O}} &:= \prod_{j \in \mathbb{O}} a_0^{n+j}a_1^{n+j}
\end{align*}
Then we have the following.
\begin{align*}
    \psi_n(f) &= m_{\FF,\bfa} \sum_{\bfg \in L(n,\FF,\bfa)} c_\bfg
                 \left( a_1^{n+1}a_1^{n+1}m_{\FF,\mathbb{O}} \prod_{j \in \mathbb{E}} (a_{\sum_{\ell = 1}^{j} g_\ell}^{n+j})^2 \right.\\   
                 &\;\;\;\;\;\;\;\;\;\;\;\;\;\;\;\;\;\;\;\;+ a_1^{n+1}a_{1}^{2n} m_{\FF,\mathbb{E}}
                 \prod_{j \in \mathbb{O}} (a_{\sum_{\ell=1}^j g_\ell}^{n+j})^2 +
                  a_1^{n+1}a_{1}^{2n} m_{\FF,\mathbb{E}} \prod_{j \in \mathbb{O}} (a_{\sum_{\ell =2}^j g_\ell}^{n+j})^2 \\    
                 &\;\;\;\;\;\;\;\;\;\;\;\;\;\;\;\;\;\;\;\;\left.+ a_{1}^{2n}a_{1}^{2n} m_{\FF,\mathbb{O}}
                 \prod_{j \in \mathbb{E}} (a_{\sum_{\ell=2}^j g_\ell}^{n+j})^2
                 \right)    \\
                 &= m_{\FF,\bfa}m_{\FF,\mathbb{O}} (a_1^{n+1})^2\sum_{\bfg \in L(n,\FF,\bfa)} c_\bfg  \prod_{j \in \mathbb{E}} (a_{\sum_{\ell = 1}^{j} g_\ell}^{n+j})^2  \\
                 &\;\;\;\;\;\;\;\;\;\;\;\;\;\;\;\;\;\;\;\;+ m_{\FF,\bfa}m_{\FF,\mathbb{O}} (a_{1}^{2n})^2\sum_{\bfg \in L(n,\FF,\bfa)} c_{\bfg}  \prod_{j \in \mathbb{E}} (a_{\sum_{\ell = 2}^{j} g_\ell}^{n+j})^2  \\
                 &\;\;\;\;\;\;\;\;\;\;\;\;\;\;\;\;\;\;\;\;+ m_{\FF,\bfa}m_{\FF,\mathbb{E}} a_1^{n+1}a_1^{2n} \sum_{\bfg \in L(n,\FF,\bfa)} c_{\bfg}  \prod_{j \in \mathbb{O}} (a_{\sum_{\ell=1}^j g_\ell}^{n+j})^2 \\
                 &\;\;\;\;\;\;\;\;\;\;\;\;\;\;\;\;\;\;\;\;+ m_{\FF,\bfa}m_{\FF,\mathbb{E}} a_1^{n+1}a_1^{2n}\sum_{\bfg \in L(n,\FF,\bfa)} c_{\bfg}  \prod_{j \in \mathbb{O}} (a_{\sum_{\ell=2}^j g_\ell}^{n+j})^2 \\
\end{align*}


In the final expression of the equation above, there are four sums. The monomials in the first two sums have the same superscripts, and the monomials in the second two sums have the same superscripts. Moreover, these two sets of supersctipts are disjoint, so there can be no cancellation among these pairs of sums. Thus, \(\psi_n(f) = 0\) if and only if the following equations hold.
\begin{align}
    0 &= (a_1^{n+1})^2\sum_{\bfg \in L(n,\FF,\bfa)} c_\bfg  \prod_{j \in \mathbb{E}} (a_{\sum_{\ell = 1}^{j} g_\ell}^{n+j})^2 + (a_{1}^{2n})^2\sum_{\bfg \in L(n,\FF,\bfa)} c_{\bfg}  \prod_{j \in \mathbb{E}} (a_{\sum_{\ell = 2}^{j} g_\ell}^{n+j})^2 \\
    0 &= \sum_{\bfg \in L(n,\FF,\bfa)} c_{\bfg}  \prod_{j \in \mathbb{O}} (a_{\sum_{\ell=1}^j g_\ell}^{n+j})^2 + \sum_{\bfg \in L(n,\FF,\bfa)} c_{\bfg}  \prod_{j \in \mathbb{O}} (a_{\sum_{\ell=2}^j g_\ell}^{n+j})^2
\end{align}
In (1), there can be no cancellation among these two sums because of the coefficients \((a_1^{n+1})^2\) and \((a_1^{2n})^2\) in front of the sums. The subscripts in each of these sums are all off by exactly 1; therefore, the first term is 0 if and only if the second term is 0. In (2), the subscripts in each sum are also again off by exactly 1. In order to show there is no cancellation among these sums, we will show that the monomials appearing in each sum are distinct.

\begin{lemma}
There are no distinct \(\bfg,\bfg' \in L(n,\FF,\bfa)\) so that \(\sum_{\ell=1}^j g_\ell = \sum_{\ell=2}^j g_\ell'\) for all \(2 \leq j \leq n-1\) so that \([j] \sm \FF\) has odd cardinality. In other words, in (2), the monomials in the two sums above are disjoint.
\end{lemma}
\begin{proof}
Let \(\{i_1,\dotsc, i_{m}\} = [n-1]\sm\FF\). Suppose \(\bfg,\bfg'\in L(n,\FF,\bfa)\) and \(\sum_{\ell = 1}^j g_\ell = \sum_{\ell = 2}^j g_\ell'\) for all \(j \in \{i_1,\dotsc,i_m\}\). Since \(g_1' = 1\), we have \(\sum_{\ell = 1}^j g_\ell = 1 + \sum_{\ell = 1}^j g_\ell'\) for all \(j \in \{i_1,\dotsc,i_m\}\). Since \(\bfg\vert_\bfa = \bfg'\vert_\bfa\), we see that \(g_{i_1} = 1 + g_{i_1}'\). However, this contradicts that \(\bfg' \in L(n,\FF,\bfa)\). Since \(L(n,\FF,\bfa) = \{\bfg ~|~ q_\bfg q_\bfh \in \mathcal{G}(n, \FF,\bfa) \text{ and } \bfg <_{lex} \bfh\}\),  there is some \(\bfh'\) so that \(q_{\bfg'}q_{\bfh'} \in \mathcal{G}(n, \FF,\bfa)\), and since \(i_1 \notin \FF\), \(h_{i_1}' = 0\) which implies \(\bfh' <_\text{lex} \bfg'\) and \(\bfg' \notin L(n,\FF,\bfa)\).
\end{proof}

All this is to show that equations (1) and (2) reduce to the following equations. Thus, \(\psi_n(f) = 0\) if and only if the following equations hold.
\begin{align*}
    0 &= \sum_{\bfg \in L(n,\FF,\bfa)} c_\bfg 
         \prod_{j \in \mathbb{E}} (a_{\sum_{\ell =1}^j g_\ell}^{n+j})^2 \\
    0 &= \sum_{\bfg \in L(n,\FF,\bfa)} c_\bfg 
         \prod_{j \in \mathbb{O}} (a_{\sum_{\ell = 1}^j g_\ell}^{n+j})^2
\end{align*}
Recalling the definitions of \(M_\mathbb{E}^{n, \FF,\bfa}\) and \(M_\mathbb{O}^{n, \FF,\bfa}\), we see that \(f \in J_n \cap \cg(n,\FF,\bfa)\) if and only if it lies in the intersection of \(\ker M_\mathbb{E}^{n,\FF,\bfa}\) and \(\ker M_\mathbb{O}^{n, \FF,\bfa}\).

{\bf Case 3: \(1\in\FF\) and \(a_1 =0\).} Note that for each monomial, \(q_\bfg q_\bfh\), in \(\mathcal{G}(n, \FF,\bfa)\), \(g_1\) and \(h_1\) are 0. We will always assume that \(\bfg \in L(n,\FF,\bfa)\).Now, we compute \(\psi_n(f)\). 
\begin{align*}
    \psi_n(f) &= \sum_{\bfg \in L(n,\FF,\bfa)} c_\bfg \psi_n (q_\bfg q_\bfh) \\
                 &= \sum_{\bfg \in L(n,\FF,\bfa)} c_\bfg \left(\prod_{j=1}^n a_{g_j}^j \right)\left(\prod_{j=1}^{n-1} a_{\sum_{\ell =1}^{j} g_\ell}^{n+j}  + \prod_{j=2}^{n} a_{\sum_{\ell =2}^{j} g_\ell}^{n+j}\right) \left( \prod_{j=1}^n a_{h_j}^j \right) \left(\prod_{j=1}^{n-1} a_{\sum_{\ell =1}^{j} h_\ell}^{n+j}  + \prod_{j=2}^{n} a_{\sum_{\ell =2}^{j} h_\ell}^{n+j} \right)\\
                 &= \sum_{\bfg \in L(n,\FF,\bfa)} c_\bfg 
                 \left(\prod_{j \in \FF} (a_{g_j}^j)^2\right)
                 \left(\prod_{j \notin \FF} a_0^j a_1^j \right)
                 \left(\prod_{j=1}^{n-1} a_{\sum_{\ell =1}^{j} g_\ell}^{n+j}  + \prod_{j=2}^{n} a_{\sum_{\ell =2}^{j} g_\ell}^{n+j}\right)
                 \left(\prod_{j=1}^{n-1} a_{\sum_{\ell =1}^{j} h_\ell}^{n+j}  + \prod_{j=2}^{n} a_{\sum_{\ell =2}^{j} h_\ell}^{n+j}\right)
\end{align*}
The monomial \(\left(\prod_{j \in \FF} (a_{g_j}^j)^2\right)\left(\prod_{j \notin \FF} a_0^j a_1^j \right)\) depends only on \(\FF\) and \(\bfa\), so we note this can be factored out and we denote it by \(m_{\FF,\bfa}\).
\begin{align*}
   \psi_n(f) &= m_{\FF,\bfa} \sum_{\bfg \in L(n,\FF,\bfa)} c_\bfg 
                \left(\prod_{j=1}^{n-1} a_{\sum_{\ell =1}^{j} g_\ell}^{n+j}  + \prod_{j=2}^{n} a_{\sum_{\ell =2}^{j} g_\ell}^{n+j}\right)
                \left(\prod_{j=1}^{n-1} a_{\sum_{\ell =1}^{j} h_\ell}^{n+j}  + \prod_{j=2}^{n} a_{\sum_{\ell =2}^{j} h_\ell}^{n+j}\right) \\
                &= m_{\FF,\bfa} \sum_{\bfg \in L(n,\FF,\bfa)} c_\bfg 
                \left(a_{g_1}^{n+1}\prod_{j=2}^{n-1} a_{\sum_{\ell =1}^{j} g_\ell}^{n+j}  + a_{g_1}^{2n}\prod_{j=2}^{n-1} a_{\sum_{\ell =2}^{j} g_\ell}^{n+j}\right)
                \left(a_{h_1}^{n+1}\prod_{j=2}^{n-1} a_{\sum_{\ell =1}^{j} h_\ell}^{n+j}  + a_{h_1}^{2n}\prod_{j=2}^{n-1} a_{\sum_{\ell =2}^{j} h_\ell}^{n+j}\right) \\
                &= m_{\FF,\bfa} \sum_{\bfg \in L(n,\FF,\bfa)} c_\bfg 
                \left(a_{0}^{n+1}\prod_{j=2}^{n-1} a_{\sum_{\ell =1}^{j} g_\ell}^{n+j}  + a_{0}^{2n}\prod_{j=2}^{n-1} a_{\sum_{\ell =2}^{j} g_\ell}^{n+j}\right)
                \left(a_{0}^{n+1}\prod_{j=2}^{n-1} a_{\sum_{\ell =1}^{j} h_\ell}^{n+j}  + a_{0}^{2n}\prod_{j=2}^{n-1} a_{\sum_{\ell =2}^{j} h_\ell}^{n+j}\right)
\end{align*}
Now, we will go through the tedious task of multiplying these two binomials. In order to simplify the computation, we make the following obsevations about the various sums in the subscripts.
\begin{itemize}
    \item \(\sum_{\ell = 1}^j g_\ell  = \sum_{\ell = 2}^j g_\ell\) since \(g_1 = 0\)
    \item \(\sum_{\ell = 1}^j h_\ell = \sum_{\ell = 2}^j h_\ell\) since \(h_1 = 0\)
    \item \(\sum_{\ell = 1}^j g_\ell = \sum_{\ell = 1}^j h_\ell\) if and only if \([j] \sm \FF\) has even cardinality.
    \item \(\sum_{\ell = 2}^j g_\ell = \sum_{\ell = 2}^j h_\ell\) if and only if \([j] \sm \FF\) has even cardinality.
\end{itemize}

With these observations, we get the following:
\begin{align*}
    \psi_n(f) &= m_{\FF,\bfa} \sum_{\bfg \in L(n,\FF,\bfa)} c_\bfg
                 \left( a_0^{n+1}a_0^{n+1} \prod_{j \in \mathbb{E}} (a_{\sum_{\ell = 1}^{j} g_\ell}^{n+j})^2 
                 \prod_{j \in \mathbb{O}} a_0^{n+j}a_1^{n+j} \right.\\   
                 &\;\;\;\;\;\;\;\;\;\;\;\;\;\;\;\;\;\;\;\;+ a_0^{n+1}a_{0}^{2n} \prod_{j \in \mathbb{E}} (a_{\sum_{\ell=1}^j g_\ell}^{n+j})^2 
                 \prod_{j \in \mathbb{O}} a_0^{n+j} a_1^{n+j}  \\
                 &\;\;\;\;\;\;\;\;\;\;\;\;\;\;\;\;\;\;\;\;+ a_0^{n+1}a_{0}^{2n} \prod_{j \in \mathbb{E}} (a_{\sum_{\ell =1}^j g_\ell}^{n+j})^2 
                 \prod_{j \in \mathbb{O}} a_0^{n+j} a_1^{n+j} \\    
                 &\;\;\;\;\;\;\;\;\;\;\;\;\;\;\;\;\;\;\;\;\left.+ a_{0}^{2n}a_{0}^{2n} \prod_{j \in \mathbb{E}} (a_{\sum_{\ell=1}^j g_\ell}^{n+j})^2
                 \prod_{j \in \mathbb{O}} a_0^{n+j}a_1^{n+j}
                 \right) 
\end{align*}
Note that the product \(\prod_{j \in \mathbb{O}} a_0^{n+j}a_1^{n+j}\) depends only on \(\FF\) and \(\mathbb{O}\), so we set it equal to \(m_{\FF,\mathbb{O}}\). Then we have the following.

\begin{align*}
    \psi_n(f) &= m_{\FF,\bfa} m_{\FF,\mathbb{O}}(a_0^{n+1} + a_0^{2n})^2\sum_{\bfg \in L(n,\FF,\bfa)} 
                   c_\bfg \prod_{j \in \mathbb{E}} (a_{\sum_{\ell=1}^j g_\ell}^{n+j})^2
\end{align*}

Recalling the definition of \(M_\mathbb{E}^{n, \FF,\bfa}\), we see that \(\psi_n(f) = 0\) if and only if \(f \in \ker M_\mathbb{E}^{n,\FF,\bfa}\).

\section{Algebraic Properties of Small Sunlet Networks}
\label{sec:SmallSunlets}

\subsection{The $4$-Sunlet Network}
\label{sec:4Sunlet}

In this section, we use a toric initial ideal of $J_4$ to show that $\cs_4$ is normal and Gorenstein.

We consider a monomial weighting ${\bf w} = (w_{0000}, w_{1111}, w_{0011}, w_{1100}, w_{0101}, w_{1010}, w_{0110}, w_{1001})$, $w_{ijkl} \in \Zz$ of the generators of the polynomial ring $R_4 = \Cc[q_{0000}, q_{1111}, q_{0011},q_{1100}, q_{0101}, q_{1010},q_{0110},q_{1001}]$ which satisfies the following equalities and inequalities:

\[w_{0000} + w_{1111} = w_{0011} + w_{1100} > w_{0101} + w_{1010}, w_{0110} + w_{1001}\]

\noindent
The associated initial ideal of $J_4 = \langle q_{0000}q_{1111} - q_{0011}q_{1100} + q_{0101}q_{1010} - q_{0110}q_{1001}\rangle$ is generated by the binomial $q_{0000}q_{1111} - q_{0011}q_{1100}$.  

\begin{definition}
Let $\Delta_4 \subset \Rr^6$ be the convex hull of the points $(0, 0, 0, 0, 0, 0)$, $(1, 0, 0, 0, 0, 0)$, $(0, 1, 0, 0, 0, 0)$, $(1, 1, 0, 0, 0, 0)$, $(0, 0, 1, 0, 0, 0)$, $(0, 0, 0, 1, 0, 0)$, $(0, 0, 0, 0, 1, 0)$, and $(0, 0, 0, 0, 0, 1)$. Let $G_4 \subset \Zz^{6 + 1}$ be the graded semigroup obtained by taking the integral points in the cone $P_4 \subset \Rr^{6+1}$ over $\Delta_4 \times \{1\} \subset \Rr^{6+1}$.
\end{definition}

\begin{proposition}
The initial algebra $R_4/ in_{\bf w}(J_4)$ is isomorphic to the affine semigroup algebra $\Cc[G_4]$.  The latter is normal and Gorenstein with ${\bf a}-$invariant equal to $-6$
\end{proposition}

\begin{proof}
The algebra $\Cc[G_4]$ is a polynomial ring in four variables $t^{0010001}, t^{0001001}, t^{0000101}, $ and $t^{0000011}$ over the subalgebra $A = \Kk[t^{0000001}, t^{1100001}, t^{1000001}, t^{0100001}]$. The relations among the generators of the algebra $A$ are generated by the relation $t^{0000001}t^{1100001} - t^{1000001}t^{0100001}$. It follows that $\Kk[G_4]$ is normal and Gorenstein. The canonical module of $\Cc[G_4]$ is isomorphic to the ideal generated by $G_4 \cap \textup{int}(P_4)$. In turn, this ideal is principal and generated by the degree $6$ element $t^{1111116} = t^{0000001}t^{1100001}t^{0010001}t^{0001001}t^{0000101}t^{0000011}$.

We define a map $\phi: R_4 \to \Cc[G_4]$ as follows:

\[q_{0000} \to t^{0000001} \ \ \ \ q_{1111} \to t^{1100001}\]
\[q_{0011} \to t^{1000001} \ \ \ \ q_{1100} \to t^{0100001}\]
\[q_{0101} \to t^{0010001} \ \ \ \ q_{1010} \to t^{0001001}\]
\[q_{0110} \to t^{0000101} \ \ \ \ q_{1001} \to t^{0000011}\]

\noindent
The kernel of $\phi$ is seen to be $in_{\bf w}(J_4) = \langle q_{0000}q_{1111} - q_{0011}q_{1100}\rangle$.  
\end{proof}

We can compute the weight of each generator of $G_4$ along each edge of the four leaf network by mapping it to a monomial in $R_4/in_{\bf w}(J_4)$ with $\phi$. Let $\pi_i: G_4 \to \Zz_{\geq 0}e_0 + \Zz_{\geq 0}e_1$ be the map which assigns an element $u \in G_4$ the weight along the $i$-th edge.  The generator of the canonical module of $\Kk[G_4]$ corresponds to the monomial $q_{0000}q_{1111}q_{0101}q_{1010}q_{0110}q_{1001}$. This monomial has weight $3e_0 + 3e_1$ on each edge in the $4$-cycle. 

The algebra $R_4/J_4$ is multigraded by the group $(\Zz_{\geq 0}e_0 + \Zz_{\geq 0}e_1)^4$. The multigrading is shared by the degeneration $\Cc[G_4]$, where it corresponds to the linear projection 

\[\bar{\pi} = (\pi_1, \pi_2, \pi_3, \pi_4): G_4 \to (\Zz_{\geq 0}e_0 + \Zz_{\geq 0}e_1)^4.\] 

The image of $\bar{\pi}$ is the set $Q_4 \subset (\Zz_{\geq 0}e_0 + \Zz_{\geq 0}e_1)^4$ 
of $(A_1e_0 + A_2e_1, B_1e_0 + B_2e_1, C_1e_0+C_2e_1, D_1e_0+D_2e_1)$ where $A_1 + A_2 = B_1 + B_2 = C_1 + C_2 = D_1 + D_2$ and $A_1 + B_1 + C_1 + D_1 \in 2\Zz$. 
\begin{remark}
Note that the multigrading by \((\Zz_{\geq 0}e_0 + \Zz_{\geq 0}e_1)^4\) coincides with the grading by \(\Zz^5\) described in Section \ref{sec:Zgrading} by sending \((A_1e_0 + A_2e_1, B_1e_0 + B_2e_1, C_1e_0+C_2e_1, D_1e_0+D_2e_1)\) to \((A_1 + A_2, A_2, B_2, C_2, D_2)\).
\end{remark}
Fix $p \in Q_4$, then the number of elements of $G_4$ which map to $p$ under $\bar{\pi}$ coincides with the value $h_{R_4/J_4}(p)$ of the multigraded Hilbert function of $R_4/J_4$. This value can be computed as follows. Let $A' = A_1 + \MIN\{0, \frac{1}{2}(C_1 + D_1 - A_1 -B_1)\}$, $B' = B_1 + \MIN\{0, \frac{1}{2}(C_1 + D_1 - A_1 -B_1)\}$, $C' = C_1 - \MIN\{0, \frac{1}{2}(C_1 + D_1 - A_1 -B_1)\}$, $D' = D_1 - \MIN\{0, \frac{1}{2}(C_1 + D_1 - A_1 -B_1)\}$, and $E' = A_1 + A_2 + \MIN\{-A_1 - B_1,  - C_1 - D_1\}$, then 

\[h_{R_4/J_4}(p) = \frac{1}{2}(\MIN\{A',B',C',D'\} +\MIN\{0, E'\} +1)(3\MIN\{A',B',C',D'\} -\MIN\{0, E'\} +2).\]

The Hilbert series is given by 
\[H_{R_4/J_4}(T) = \frac{1+T}{\left({1-T}\right)^{7}}.\]

Now fix a $4$-valent tree $\mathcal{T}$, and let $N$ be the network optained by gluing $4$-sunlet networks together according to $\mathcal{T}$.  Let $G_\mathcal{T}$ be the toric fiber product of $E(\mathcal{T})$ according to the topology of $\mathcal{T}$.  The next proposition establishes the basic properties of the semigroup algebra $\Cc[G_\mathcal{T}]$ and the network algebra $\Cc[q]/I_N$.

\begin{proposition}
The semigroup $G_\mathcal{T}$ is generated in degree $1$.  Its generators are the lattice points in a normal polytope $\Delta_\mathcal{T}$ obtained as a fiber product polytope of $E(\mathcal{T})$ copies of $\Delta_4$ over the topology of $\mathcal{T}$.  With these generators, the semigroup algebra $\Cc[G_\mathcal{T}]$ is presented by a quadratic ideal, and is Gorenstein with ${\bf a}-$invariant equal to $-6$. Moreover, the algebra $\Cc[q]/I_N$ is normal, presented by quadratics, and Gorenstein with ${\bf a}$-invariant equal to $-6$, and its Hilbert function agrees with Ehrhart polynomial of $\Delta_\mathcal{T}$.
\end{proposition}

\begin{proof}
This is a consequence of Propositions \ref{prop-toricfibernormal} and \ref{prop-toricfibergor}.
\end{proof}

\subsection{The 5-Sunlet Network}
\label{exa:5sunlet}
In this section, we focus on the 5-sunlet network $\cs_5$ and its  corresponding ideal $J_5$. We describe the structure of its generating set and also discuss some properties of the ideal. All computations for this section can be found in the macaulay2 file \texttt{sunlet5.m2}.

We first computed the ideal $J_5$ by elimination with a degree bound. We computed a Gr\"obner basis for the elimination ideal up to degree 2 and then verified that the result was prime and of the correct dimension which is 10. The dimension is obtained by computing the rank of the Jacobian of $\psi_{\cs_5}$ symbolically. As a result we get that

\begin{align*}
J_5 =
\langle
& q_{10111}q_{11000} - q_{10100}q_{11011} + q_{10010}q_{11101} - q_{10001}q_{11110}, \\
& q_{01111}q_{11000} - q_{01100}q_{11011} + q_{01010}q_{11101} - q_{01001}q_{11110}, \\
& q_{01111}q_{10100} - q_{01100}q_{10111} + q_{00110}q_{11101} - q_{00101}q_{11110}, \\
& q_{01111}q_{10010} - q_{01010}q_{10111} + q_{00110}q_{11011} - q_{00011}q_{11110}, \\
& q_{01100}q_{10010} - q_{01010}q_{10100} + q_{00110}q_{11000} - q_{00000}q_{11110}, \\
& q_{01111}q_{10001} - q_{01001}q_{10111} + q_{00101}q_{11011} - q_{00011}q_{11101}, \\
& q_{01100}q_{10001} - q_{01001}q_{10100} + q_{00101}q_{11000} - q_{00000}q_{11101}, \\
& q_{01010}q_{10001} - q_{01001}q_{10010} + q_{00011}q_{11000} - q_{00000}q_{11011}, \\
& q_{00110}q_{10001} - q_{00101}q_{10010} + q_{00011}q_{10100} - q_{00000}q_{10111}, \\
& q_{00011}q_{01100} - q_{00000}q_{01111}, \\
& q_{00110}q_{01001} - q_{00101}q_{01010}
\rangle.
\end{align*}

We also computed the tropical variety explicitly. It has 252 maximal cones. Using \texttt{sunlet5.m2}, we found that 116 of these maximal cones give prime toric initial ideals \cite{M2}. The toric varieties corresponding to these 116 cones are all normal which was checked using normaliz. The following example showcases one of these toric degenerations. 

\begin{example}
\label{ex:5-sunletdegen}
Consider the weight vector 
\begin{align*}
{\bf w}   &= (w_{00000}, w_{00011}, w_{00101}, w_{00110}, w_{01001}, w_{01010}, w_{01100}, w_{01111},  \\
    & \hspace{.5cm} w_{10001}, w_{10010}, w_{10100}, w_{10111}, w_{11000}, w_{11011}, w_{11101}, w_{11110})\\
    &= (0, 0, 0, 0, 0, 0, -3, -3, 0, -2, -3, -3, -4, -4, -4, -5)
\end{align*}
Using gfan, we found that with respect to this weight vector, the polynomials in the left column form a Gr\"obner basis for \(J_5\), and the terms with the {\it lowest} weights are underlined. The polynomials in the right column are the corresponding initial forms which generate \(in_{\bf w}(J_5)\).
\begin{align*}
 \underline{q_{10111}q_{11000}} - \underline{q_{10100}q_{11011}} + q_{10010}q_{11101} - q_{10001}q_{11110} & & q_{1 0 1 1 1}q_{1 1 0 0 0}-q_{1 0 1 0 0}q_{1 1 0 1 1}\\
 \underline{q_{01111}q_{11000}} - \underline{q_{01100}q_{11011}} + q_{01010}q_{11101} - q_{01001}q_{11110} & & q_{0 1 1 1 1}q_{1 1 0 0 0}-q_{0 1 1 0 0}q_{1 1 0 1 1}\\
 \underline{q_{01111}q_{10100}} - \underline{q_{01100}q_{10111}} + q_{00110}q_{11101} - q_{00101}q_{11110} & & q_{0 1 1 1 1}q_{1 0 1 0 0}-q_{0 1 1 0 0}q_{1 0 1 1 1} \\
 \underline{q_{01111}q_{10010}} - q_{01010}q_{10111} + q_{00110}q_{11011} - \underline{q_{00011}q_{11110}} & & q_{0 1 1 1 1}q_{1 0 0 1 0}-q_{0 0 0 1 1}q_{1 1 1 1 0}\\
 \underline{q_{01100}q_{10010}} - q_{01010}q_{10100} + q_{00110}q_{11000} - \underline{q_{00000}q_{11110}} & & q_{0 1 1 0 0}q_{1 0 0 1 0}-q_{0 0 0 0 0}q_{1 1 1 1 0}\\
 q_{01111}q_{10001} - q_{01001}q_{10111} + \underline{q_{00101}q_{11011}} - \underline{q_{00011}q_{11101}} & & q_{0 0 1 0 1}q_{1 1 0 1 1}-q_{0 0 0 1 1}q_{1 1 1 0 1}\\
 q_{01100}q_{10001} - q_{01001}q_{10100} + \underline{q_{00101}q_{11000}} - \underline{q_{00000}q_{11101}} & & q_{0 0 1 0 1}q_{1 1 0 0 0}-q_{0 0 0 0 0}q_{1 1 1 0 1}\\
 q_{01010}q_{10001} - q_{01001}q_{10010} + \underline{q_{00011}q_{11000}} - \underline{q_{00000}q_{11011}} & & q_{0 0 0 1 1}q_{1 1 0 0 0}-q_{0 0 0 0 0}q_{1 1 0 1 1} \\
 q_{00110}q_{10001} - q_{00101}q_{10010} + \underline{q_{00011}q_{10100}} - \underline{q_{00000}q_{10111}} & & q_{0 0 0 1 1}q_{1 0 1 0 0}-q_{0 0 0 0 0}q_{1 0 1 1 1}\\
 \underline{q_{00011}q_{01100}} - \underline{q_{00000}q_{01111}}  & & q_{0 0 0 1 1}q_{0 1 1 0 0}-q_{0 0 0 0 0}q_{0 1 1 1 1}\\
 \underline{q_{00110}q_{01001}} - \underline{q_{00101}q_{01010}} & & q_{0 0 1 1 0}q_{0 1 0 0 1}-q_{0 0 1 0 1}q_{0 1 0 1 0}
\end{align*}


The ideal, \(in_{\bf w}(J_5)\), defines a toric variety which is parameterized by monomials whose exponent vectors are the columns in the matrix below.  This matrix was found using \cite[Theorem 4]{Kaveh-Manon-NOK}.  In particular, the Fourier coordinate generators of $S_5$ are a \emph{Khovanskii basis} of a valuation associated to the cone containing ${\bf w}$, and the convex hull of the columns of $A$ in $\Rr^{16}$ is a \emph{Newton-Okounkov} body of the sunlet variety $V_5 \subset \mathbb{P}^{15}$.

\[
 A = {\left({\begin{array}{cccccccccccccccc}
      1&1&1&1&1&1&1&1&1&1&1&1&1&1&1&1\\
      0&0&1&1&3&3&0&0&1&1&2&2&0&0&1&1\\
      0&0&1&3&1&3&0&0&2&0&1&1&1&1&2&0\\
      0&0&0&2&0&2&2&2&6&0&0&0&0&0&0&2\\
      0&0&2&2&2&2&0&0&2&2&0&0&0&0&2&2\\
      1&1&1&1&0&0&0&0&1&1&1&1&0&0&0&0\\
      2&1&1&1&1&1&1&0&1&1&1&0&1&0&0&0\\
      0&1&0&0&1&1&0&1&1&1&0&1&1&2&1&1\\
      0&0&1&0&1&0&1&1&1&0&1&1&1&1&2&1\\
      0&0&0&1&0&1&1&1&0&1&1&1&1&1&1&2\\
      0&0&0&0&0&0&0&0&1&1&1&1&1&1&1&1\\
      0&0&0&0&1&1&1&1&0&0&0&0&1&1&1&1\\
      0&1&1&1&1&1&1&2&1&1&1&2&1&2&2&2\\
      2&1&2&2&1&1&2&1&1&1&2&1&1&0&1&1\\
      2&2&1&2&1&2&1&1&1&2&1&1&1&1&0&1\\
      2&2&2&1&2&1&1&1&2&1&1&1&1&1&1&0\\
      \end{array}}\right)}
      \]

Using normaliz, we were able to show that the semigroup generated by the columns of \(A\) is saturated with respect to the rank 10 sublattice of \(\Zz^{16}\) that they span; hence, this is a normal toric variety from which we can conclude that \(\cs_5\) is normal and Cohen-Macaulay. Moreover, the Hilbert series is given by
\[
H_{R_5/J_5}(T) = \frac{1 + 6T + 10T^2 + 6T^3 + T^4}{(1-T)^{10}}.
\]


Since the numerator is symmetric and since it is Cohen-Macaulay, \cite[Theorem 4.4]{Stanley} shows that \(\cs_5\) is Gorenstein. These computations can be found in \texttt{sunlet5.m2}. One can also check \(\cs_5\) is Gorenstein by noting that that the canonical module of \(\Cc[\Nn A] \isom R_5/in_{\bf w}(J_5)\) is generated by the following vector:
 \[
      (6,8,8,10,8,3,4,5,5,5,4,3,8,7,7,7)^t.
 \]
This exponent vector corresponds to the degree 6 monomial \(q_{00000}q_{01010}q_{10001}q_{10111}q_{11101}q_{11110}\).
\end{example}

In this last proposition, we record all the algebraic properties of \(\cs_5\) that we investigated in the previous example, and we record that level-1 networks built from 4- and 5-sunlets are Cohen-Macaulay

\begin{proposition}
\(\cs_5\) is a normal, Gorenstein variety. Its tropicalization has 252 maximal cones, 116 of which yield prime binomial initial ideals.
\end{proposition}

\begin{corollary}
Any level-1 network built out of 4- and 5-sunlet networks is a normal Cohen-Macualay variety.
\end{corollary}
\begin{proof}
Since 4- and 5-sunlet varieties are normal and Cohen-Macualay, combining Proposition \ref{prop-toricfibernormal} with Proposition \ref{prop-tfp} shows that any level-1 network built out from 4- and 5-sunlet networks is normal and Cohen-Macualay.
\end{proof}

\section{Open Problems}
In this section, we discuss some conjectures for which we have computational evidence and suggest some possible techniques for solving them. We also provide some interesting open problems surrounding sunlet network ideals. 

One of the main drawbacks to the techniques used in Section \ref{sec:gloves} is that it only yields quadratic generators for \(J_n\). For \(n\)-sunlet networks with \(4 \leq n \leq 7\), we have verified that their ideals are quadratically generated. This was done in Macaulay2 by showing that over \(\Qq\), \(\ker \psi_n = I_n\) for \(n = 4,5,6,\) and 7. Since we had equality over \(\Qq\), the ideals must still be equal after extending to the complex numbers. While we have verified that \(J_n\) is generated by quadratics for \(4 \leq n \leq 7\), it remains open as to whether these generate \(J_n\) for \(n \geq 8\).  For the CFN model, the ideals for trees are always generated by quadratics, and as we have seen the quadratic invariants obtained for the sunlet ideals are built from invariants from the underlying trees; hence, we suspect that \(J_n\) is always quadratically generated.

\begin{conjecture}
\label{conj:quadraticallygenerated}
Let \(I_n\) be the ideal generated by all quadratic invariants in \(J_n\). Then \(I_n = J_n\) for all \(n\geq 4\).
\end{conjecture}

In order to prove Conjecture \ref{conj:quadraticallygenerated}, it would be enough to show that \(I_n\) is prime and of the correct dimension. To this end, we have the following conjecture which would prove Conjecture \ref{conj:quadraticallygenerated}. 

\begin{conjecture}
For \(n \geq 5\), \(\dim J_n = 2n = \dim I_n\) and \(I_n\) is prime.
\end{conjecture}

A possible approach to proving that $I_n$ is prime is that taken in \cite{LS15-Strand}. The main workhorse of their technique is the following lemma which was originally stated in \cite[Proposition 23]{GSS05}.

\begin{lemma}\cite[Lemma 2.5]{LS15-Strand}
\label{lemma:primeDescent}
Let $k$ be a field and $J \subset k[x_1, \ldots x_n]$ be an ideal containing a polynomial $f = g x_1 + h$ with $g, h$ not involving $x_1$ and $g$ a non-zero divisor modulo $J$. Let $J_1 = J \cap k[x_2, \ldots x_n]$ be the elimination ideal. Then $J$ is prime if and only if $J_1$ is prime. 
\end{lemma}

This lemma can be used to create a descending chain of ideals each one involving one less variable. As long as a polynomial $f$ of the required form can be found, then one can prove that the original ideal is prime by verifying that the last ideal in the chain is prime. For $4 \leq n \leq 7$ we have done this with $I_n$ by repeatedly eliminating variables in reverse lexicographic order until we are left with an ideal in only the variables $q_\bfg$ such that $g_1 = 0$. That is we build a chain
\[
I_n \supset I_n^{(1)} \supset \cdots \supset I_n^{(2^{n-2})}
\]
where $I_n^{(j)}$ is obtained by eliminating the $j$th variable in reverse lexicographic order from $I_n^{(j-1)}$ and at each step we ensure that a polynomial $f$ of the form described in Lemma \ref{lemma:primeDescent} exists. Typically one would then need to verify that $I_n^{(2^{n-2})}$ is prime but the following lemma shows there is no need for this. Our implementation of this can be found in the macaulay2 file \texttt{primeDescent.m2}. 

\begin{lemma}
Let $I_n^{(2^{n-2})} = I_n \cap \Cc[q_\bfg : g_1 = 0]$. Then $I_n^{(2^{n-2})} \cong I_\ct$ where $\ct$ is the tree obtained by deleting the reticulation vertex of $\cs_n$ and all adjacent edges. 
\end{lemma}

This lemma implies that if one can always find a polynomial $f$ of the desired form in each of the intermediate elimination ideal $I_n^{(j)}$ then $I_n$ is prime since the last ideal $I_n^{(2^{n-2})}$ is isomorphic to a tree ideal; thus, it must be prime. 

For the question of the dimension of \(J_n\), we have the following bound. 

\begin{proposition}
For $n \geq 4$ it holds that $2n-1 \leq \dim(J_n) \leq 2n + 1$. 
\end{proposition}
\begin{proof}
First we note that $J_n$ is properly contained in the ideals $I_{\ct_0}$ and $I_{\ct_1}$ for the trees $\ct_0$ and $\ct_1$ that are obtained from $\cs_n$ by deleting reticulation edges. It is well known that each of these ideals has $\dim(I_{\ct_i}) = 2n-2$ (see for example \cite{BBND+19}). Since we have that $J_n$ is a prime ideal properly contained in these two prime ideals which are not equal, we get the lower bound $2n-1 \leq \dim(J_n)$. For the other bound recall that $V_{\cs_n}$ can also be thought of as a projective variety the map $\psi_{\cs_n}$ parameterizing $J_n$ can be thought of as a map
\[
\psi_{\cs_n}:  \prod_{e \in E(\cs_n)} \Pp^1   \to \Pp^{2^{n-1}-1}
\]

where each copy of $\Pp^1$ in the domain corresponds to an edge of $\cs_n$. This immediately implies that the projective variety corresponding to $\cs_n$ has dimension at most $\#E(\cs_n) = 2n$ and so $\dim(J_n) \leq 2n+1$.
\end{proof}

We also have that \(\dim J_n \leq \dim I_n\) as \(I_n \subseteq J_n\).
Moreover, using the rank of Jacobian of \(\psi_{\cs_n}\), we have shown for $5 \leq n \leq 8$ that the dimension of $J_n$ is $2n$. We've also computed the rank of the Jacobian with random values substituted in for the parameters for $n$ up to 17. In each case we've found that the rank is also $2n$ which means that $\dim(J_n) = 2n$ with probability 1 for $9 \leq n \leq 17$. These computations can be found in the file \texttt{sunletDim.m2}. 

As we have seen in Section \ref{sec:SmallSunlets}, the 4- and 5-sunlet networks are normal, Gorenstein varieties. We have not been able to show that \(\cs_6\) is Gorenstein; however, we have computed its Hilbert series which suggests it is indeed Gorenstein.
\[H_{R_6/J_6}(T) = \frac{1+20T+131T^{2}+376T^{3}+528T^{4}+376T^{5} + 131T^{6}+20T^{7}+T^{8}}{\left({1-T}\right)^{12}}\]
Therefore, to show that \(\cs_6\) has the Gorenstein property, it would be enough to show that it is Cohen-Macaulay by \cite[Theorem 4.4]{Stanley}.

\begin{question}
Is \(\cs_n\) normal, Cohen-Macaulay, and Gorenstein for \(n \geq 6\)?
\end{question}

In Example \ref{ex:5-sunletdegen}, we also saw that the generator of the canonical module had degree \(3e_0 + 3e_1\) for each non-reticulation leaf, while at the reticulation edge, it had degree \(2e_0 + 4e_1\). Then Propositions \ref{prop-toricfibernormal} and \ref{prop-toricfibergor} imply the following proposition.

\begin{proposition}
Let \(\mathcal{N}\) be a level-1 network obtained by gluing 4- and 5-sunlets along trees under the condition that nothing is glued to a reticulation edge in a 5-sunlet. Then the phylogenetic variety \(V_\mathcal{N}\) is Gorenstein.
\end{proposition}

The fact that the reticulation edge for a 5-sunlet has a different degree than the 4-sunlet case does not mean that other level-1 networks built out of 4- and 5-sunlets are not Gorenstein. It just means that some other proof would be needed to show the Gorenstein property.

As we have seen in Section \ref{sec:SmallSunlets}, there are very well-behaved toric degenerations of \(\cs_4\) and \(\cs_5\). In the case when \(n = 5\), there are 116 cones in the tropical variety which yield normal toric varieties; however, most of them are somewhat less well-behaved than the one shown. For example, using the weight given in Example \ref{ex:5-sunletdegen}, one sees that the quadratic invariants produced in Section \ref{sec:gloves} actually form a Gr{\"o}bner basis with respect to this weight. This is a property that does {\it not} happen for most of the weights in the tropical variety. Moreover, the initial forms of these quadratic invariants are always invariants for at least one of the underlying trees \(\ct_0\) or \(\ct_1\). To this end, we ask the following.

\begin{question}
For \(n \geq 5\), is there a weight vector \(w\) on \(R_n\) for which \(in_w(J_n)\) is a prime binomial ideal? If so, can it be shown that there is a combinatorial rule for finding such a \(w\) where a Gr{\"o}bner basis of \(J_n\) with respect to \(w\) can be deduced combinatorially?
\end{question}

This question is interesting even in the case when \(n=6\). If one was able to find a toric degeneration of \(\cs_6\) to a normal toric variety, then since the numerator of the Hilbert series is symmetric, one would also know that \(\cs_6\) is Gorenstein. 


\section*{Acknowledgments}
Joseph Cummings and Christopher Manon were partially supported by Simons Collaboration Grant (587209).  
Benjamin Hollering was partially supported by the US National Science Foundation
(DMS 1615660) and would like to thank Seth Sullivant for many helpful conversations.  
\bibliography{ref.bib}{}
\bibliographystyle{plain}

\end{document}